\documentclass[acmtog, nonacm]{acmart}

\AtBeginDocument{%
  \fancypagestyle{plain}{%
    \fancyhf{} %
    \fancyfoot[L]{ACM SIGENERGY Energy Informatics Review}%
    \fancyfoot[R]{2022}}
  \providecommand\BibTeX{{%
    \normalfont B\kern-0.5em{\scshape i\kern-0.25em b}\kern-0.8em\TeX}}}
    
\let\oldmaketitle\maketitle
\renewcommand{\maketitle}{%
  \oldmaketitle%
  \thispagestyle{plain}%
  \pagestyle{plain}}

\setcopyright{none}
\renewcommand\footnotetextcopyrightpermission[1]{}

\usepackage{tabularx,adjustbox}
\usepackage{graphicx}  
\usepackage{algorithm,algpseudocode} 
 \usepackage{longfbox}
 \usepackage{xcolor}

\usepackage{amsmath,amscd,mathrsfs}

\usepackage{amssymb}

\usepackage{array}
\usepackage{url}

\makeatletter
\newsavebox{\@brx}
\newcommand{\llangle}[1][]{\savebox{\@brx}{\(\m@th{#1\langle}\)}%
  \mathopen{\copy\@brx\kern-0.5\wd\@brx\usebox{\@brx}}}
\newcommand{\rrangle}[1][]{\savebox{\@brx}{\(\m@th{#1\rangle}\)}%
  \mathclose{\copy\@brx\kern-0.5\wd\@brx\usebox{\@brx}}}
\makeatother

\newcommand{\raf}[1]{(\ref{#1})}

\newtheorem{innercustomthm}{\textbf{Theorem}}
\newenvironment{customthm}[1]
{\renewcommand\theinnercustomthm{#1}\innercustomthm}
{\endinnercustomthm}

\usepackage{subcaption}
\begin{document}

\title{Approximately Socially-Optimal Decentralized Coalition Formation with Application to P2P Energy Sharing}

\author{Sid Chi-Kin Chau}
\email{sid.chau@anu.edu.au}
\affiliation{%
 \department{School of Computing}
  \institution{Australian National University}
  \country{Australia}
}

\author{Khaled Elbassioni}
\email{khaled.elbassioni@ku.ac.ae}
\affiliation{%
  \institution{Khalifa University}
  \country{UAE}
}

\author{Yue Zhou}
\email{yue.zhou@anu.edu.au}
\affiliation{%
 \department{School of Computing}
  \institution{Australian National University}
  \country{Australia}
}

\renewcommand{\shortauthors}{S. C.-K. Chau, K. Elbassion, Y. Zhou}

\begin{abstract}
The paradigm of P2P (peer-to-peer) economy has emerged in diverse areas. {\em P2P energy sharing} is a new form of P2P economy in the energy sector, which allows users to establish longer-term sharing arrangements of their local energy resources (e.g., rooftop PVs, home batteries) with joint optimized energy management. In such a P2P setting, a coalition of users is formed for sharing resources in a decentralized manner by self-interested users based on their individual preferences. A likely outcome of decentralized coalition formation will be a stable coalition structure, where no group of users could cooperatively opt out to form another coalition that induces higher preferences to all its members. Remarkably, there exist a number of fair cost-sharing mechanisms (e.g., equal-split, proportional-split, egalitarian and Nash bargaining solutions of bargaining games) that model practical cost-sharing applications with desirable properties, such as the existence of a stable coalition structure with a small strong price-of-anarchy (SPoA) to approximate the social optimum. In this paper, we provide general results of decentralized coalition formation: (1) We establish a logarithmic lower bound on SPoA, and hence, show several previously known fair cost-sharing mechanisms are the best practical mechanisms with minimal SPoA. (2) We show that the SPoA of egalitarian and Nash bargaining cost-sharing mechanisms to match the lower bound. (3) We derive the SPoA of a mix of different cost-sharing mechanisms. (4) We present a decentralized algorithm to form a stable coalition structure. (5) Finally, we apply our general results to P2P energy sharing and present an empirical study of decentralized coalition formation in a real-world project. We study the empirical SPoA, which is observed within $95\%$ of the social optimal cost with coalitions of 2 and 3 users, via fair cost-sharing mechanisms.
\end{abstract}





\maketitle
\pagestyle{plain}

\section{Introduction} \label{sec:intro}

People are increasingly empowered by peer-to-peer (P2P) interactions. Sharing economy is a prominent example, by which resources, services and facilities are shared dynamically among end-users in a peer-to-peer fashion.  In diverse sharing economy scenarios, users often form coalitions for sharing activities, such as ridesharing/carpooling, parking/storage sharing, and various group buying activities. In this paper, we shed light on coalition formation in sharing activities and particularly apply the P2P sharing paradigm to the application of distributed energy sharing.

\subsection{P2P Energy Sharing}

In the era of smart energy grid, end users are expected to gain increasing control of their own energy services with enhanced autonomy and transparency. The trend of decentralized control leads to the notion of ``transactive energy'', providing end users more choices and control of how energy is generated, delivered and utilized \cite{nist}. The paradigms of {\em P2P energy exchanges} \cite{TYMSPW18p2p} aim to empower transactive energy by enabling distributed coordination among local energy producers and consumers without centralized operators. The concept of P2P energy exchanges can be realized in two aspects (see Table~\ref{tbl:compare} for a comparison):
\begin{itemize}

\item {\bf Peer-to-peer Energy Trading}:
A traditional approach to enable P2P energy flows among users is by a real-time trading process that matches local suppliers and customers instantaneously \cite{CL17transenergy,LXSW14trading}. Usually, a bidding or auction process is conducted to determine the trading parties. This will result in a one-time transaction for each exchange operation that operates over a short timescale (e.g., hours). Each party is designated as a buyer (i.e., energy importer) or a seller (i.e., energy exporter). There is no joint optimization of mutual energy management systems between the parties.

\item {\bf Peer-to-peer Energy Sharing}: In contrast with real-time P2P trading, another approach is to rely on a longer-term sharing arrangement of mutual local energy systems of users \cite{chau19p2penergy}. The involved parties will agree on how to exchange energy during the agreed period (e.g., over months or years) with specifications of cost-sharing and coordination mechanisms. The longer arrangement period allows joint optimization of long-term energy management operations, where the involved parties can be hybrid prosumers, as energy importers and exporters at different times. Consequently, users are able to coordinate their consumption usage to maximize the benefits of their distributed energy generation. Sharing arrangements can be established by a negotiation and matching process. 

\end{itemize}

 \begin{figure}[t] 
        \centering \includegraphics[width=0.45\textwidth]{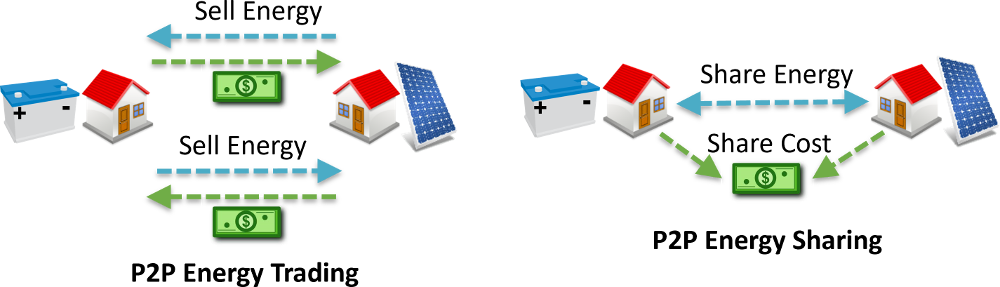} \vspace{-5pt}
        \caption{P2P energy trading vs. P2P energy sharing.}
        \label{fig:p2p} \vspace{-10pt}
    \end{figure}

\begin{table}[ht]   	
    \begin{tabularx}{\linewidth}{@{}c@{}|@{}c@{}|@{}c@{}}
    \hline    \hline
    & P2P Energy Trading & P2P Energy Sharing \\
    \hline
     Time-scale & Short (e.g., hours) & Long (e.g., months/yrs) \\
	 Roles & Designated sellers/buyers & Hybrid prosumers \\	 
	 Operations & One-time operations & Multi-time operations \\	 
	 Mechanisms & Bidding \& auction &  Negotiation \& matching \\	 	 
Joint optimization & No & Yes\\	 	 	 
    \hline    \hline
    \end{tabularx}
	\centering\caption{\label{tbl:compare}Comparison of P2P energy trading and sharing.} 
\end{table}

The concept of P2P energy sharing belongs to the paradigm of P2P sharing. Unlike the one-time energy trading,  P2P energy sharing allows users to coordinate their usage and optimize their energy management. This not only creates incentives for renewable energy adoption but also improves long-term energy efficiency. P2P energy sharing can be realized in a community for co-located users, or via virtual net metering for geographically dispersed users. In this paper, we model coalition formation in sharing activities and apply the results to P2P energy sharing.

\subsection{Modeling Decentralized Coalition Formation}

In this paper, we first study a general framework of coalition formation by self-interested users, which is subsequently applied to P2P energy sharing. Traditional cooperative game theory \cite{coopgamebook} usually considers forming a grand coalition involving every user mediated by a centralized planner, whereas practical sharing economy is often populated by small-group coalition formation in a decentralized manner by self-interested users. A better model for understanding decentralized coalition formation can be represented by a {\em hedonic game} \cite{HVW15, AB12,DG80,BJ02}, where participants aim to form coalitions among themselves according to certain preferences over the coalitions they may belong to. In such a setting, a likely outcome is a {\em stable} coalition structure of disjoint coalitions of participants, where no group of participants  could cooperatively opt out to form another possible coalition that could induce higher preferences to all its members. However, a hedonic game with arbitrary preferences may not induce a desirable stable coalition structure. 

In this paper, we adopt an approach from algorithmic game theory. We aim to devise practically useful mechanisms for decentralized coalition formation in a hedonic game that can lead to desirable stable coalition structures. In particular, we draw on our recent work \cite{CE17sharing} in studying decentralized coalition formation for the purpose of cost-sharing in sharing economy.  We previously showed some remarkable properties of several fair cost-sharing mechanisms (e.g., equal-split, proportional-split, egalitarian and Nash bargaining solutions of bargaining games) for practical cost-sharing applications, such as the existence of a stable coalition structure and a small {\em strong price-of-anarchy} (SPoA). SPoA is a common metric in algorithmic game theory \cite{AGTbook}, by which one compares the worst-case social cost of a strong Nash equilibrium (that allows any group of users to deviate jointly to form a coalition) and the cost of a social optimum. A decentralized mechanism with a small SPoA will guarantee a good approximation of the social optimum, without relying on centralized planning.  

{\bf Contributions:} In this paper, we provide several new results of decentralized coalition formation: 
\begin{enumerate}

\item We establish a logarithmic lower bound on SPoA for practical polynomial-time mechanisms, and hence, establish that several previously known fair cost-sharing mechanisms  are the best practical mechanisms with minimal SPoA.

\item We improve the SPoA of egalitarian and Nash bargaining cost-sharing mechanisms from the previously known square-root bound on SPoA to match the logarithmic lower bound.

\item We consider a mix of different cost-sharing mechanisms. We derive a logarithmic bound on SPoA with respect to mixed cost-sharing mechanisms.

\item We present a decentralized algorithm to form a stable coalition structure under several fair cost-sharing mechanisms, based on the deferred-acceptance algorithm.

\item Finally, we apply our general results of decentralized coalition formation to P2P energy sharing. We present and analyze an empirical study of decentralized coalition formation in a real-world P2P energy sharing project. We study the empirical SPoA and its effect with respect to various parameters in P2P energy sharing.
\end{enumerate}

\section{Related Work} \label{sec:related}

\subsection{Coalition Formation and Game Theory}
Coalition formation has been studied in traditional cooperative game theory \cite{coopgamebook}, which provides a foundation of forming a grand coalition, involving all the members, subject to certain axiomatic properties. For example, the notions of core, Shapley value, and nucleolus have been devised to construct proper transfer functions of utility among participants in a grand coalition \cite{coopgamebook}. However, such a setting is usually considered from the perspective of a centralized planner, who divides the benefits of a coalition among the members. 

There are a number of studies about non-cooperative coalition formation. These coalition formation models belong to the topic of hedonic games and network cost-sharing games (e.g., \cite{albers08ndg,AB12,ADTW08ndg,ADKTWR09ndg,EFM09con,HVW15,H13,RS09,RS14costsharing,KS15sharinggame}). A useful notion to model non-cooperative coalition formation is a {\em strong Nash equilibrium}. Unlike a typical Nash equilibrium that tolerates only unilateral strategic change by one participant at a time, a strong Nash equilibrium can tolerate collective strategic changes by any group of participants. Such collective strategic changes can model an alternate coalition formation from the existing coalition structure. 

Our study of decentralized coalition formation follows the approach of hedonic games \cite{AB12,H13}, where users form coalitions based on individual preferences over coalition formation.  We characterize the properties of strong Nash equilibrium in specific settings of hedonic games. But typical hedonic games allow arbitrary preferences of coalitions, which can lead to the absence of a strong Nash equilibrium (or so-called a core-stable coalition structure). In this paper, we consider hedonic games specifically for cost-sharing applications, such that users aim to split the associated cost of a coalition. In this cost-sharing setting, the individual preferences over coalition formation are governed by a cost-sharing mechanism. Our previous work \cite{CE17sharing}  showed that certain cost-sharing mechanisms can yield the existence of a strong Nash equilibrium. 


\subsection{P2P Energy}

This paper explores P2P energy sharing. Relating to P2P energy sharing, there have been an extensive body of literature on P2P energy trading. For example, see an extensive survey in \cite{TYMSPW18p2p}, and the references therein. Among these studies,  P2P energy trading has been applied to distributed energy resource management \cite{LXSW14trading, ZWCZL16bidding, TCYHSPY16sharing}. Moreover, the idea of shared pool of energy has been proposed in various energy sharing applications, such as virtual power plants \cite{MK17vpp}, energy storage cloud \cite{LZKKX17cloud, KWPV17sharing} and other multi-user energy systems \cite{WKPV16sharing, LSRI18vsolar, CSAK12storage}. P2P energy trading has been recently demonstrated in a real-world microgrid of the Brooklyn Microgrid Project \cite{MGRKOW18brooklyn}. However, a key difference between this work and previous studies is that we focus on economic mechanisms for energy sharing in the form of a long-term arrangement and coalition formation mechanisms of self-interested users.

Virtual net metering \cite{vnm20, vnmPGnE} is a bill crediting system for community solar and shared energy storage without a common electricity meter. Originally, net metering allows a single user to earn credits of the net energy export from rooftop solar or energy storage as to offset the future energy consumption. Virtual net metering can be extended to a community scenario (of geographically dispersed users) with no physical grid connection to solar and energy storage systems. In this case, the utility operator will handle the credit transfer process from an energy generation source to a group of registered users with different electricity meters. Virtual net metering can also be applied to P2P energy sharing, where the users specify the credit transfer process according to a cost-sharing mechanism.

This paper extends our prior work \cite{chau19p2penergy, chau20coalition} in P2P energy sharing. But this work provides rigorous proofs of the theoretical results and extended simulation results. Previously, we have also conducted empirical studies of decentralized coalition formation for ride-sharing \cite{chau20rideshare} and mobile edge computing \cite{chau20infocom, XZC22edge}.

\section{Model and Notations}

\subsection{General Coalition Formation Model} 

First, we present a general model of decentralized coalition formation by self-interested participants, based on hedonic games \cite{AB12,HVW15}. In Sec.~\ref{sec:p2penergy}, we will apply the model to P2P energy sharing.

In this model, there are a set of $n$ self-interested participants ${\mathcal N}$. A coalition of participants is represented by a subset $G \subseteq {\mathcal N}$. A {\em coalition structure} represents a feasible state of coalition formation, which is denoted by a partition of ${\mathcal N}$ as ${\mathcal P} \subset 2^{\mathcal N}$, such that $\bigcup_{G \in \mathcal P} G = {\mathcal N}$ and $G_1 \cap G_2 = \varnothing$ for any pair $G_1, G_2 \in {\mathcal P}$.  

Each element $G \in \mathcal P$ is called a coalition (or a group). The set of singleton coalitions, ${\mathcal P}_{\rm self} \triangleq \{\{i\} : i \in {\mathcal N} \}$, is called the {\em standalone coalition structure}, wherein no one forms a coalition with others.  We consider $K$-coalition structures with at most $K$ participants per coalition. In practice, $K$ is often much less than $n$, which models small-group coalition formation.  Let the set of partitions of ${\mathcal N}$ be ${\mathscr P}$. Let ${\mathscr P}^K \triangleq \{{\mathcal P} \in {\mathscr P} : |G| \le K \mbox{\ for each\ } G \in {\mathcal P}\}$ be the set of feasible coalition structures, such that each coalition consists of at most $K$ participants. In general, a hedonic game is defined by a complete and transitive preference relation ${\displaystyle \succcurlyeq _{i}}$ over the set ${\displaystyle \{G\subseteq {\mathcal N}:~i\in G\}}$ of coalitions that participant $i$ belongs to, but in this paper will consider only preferences defined by certain utility functions as given in the next subsection.

\subsection{Stable Coalition Structures}

Since participants are self-interested, they are motivated to join a coalition to maximize their own utility. Let $u_{i}(G)$ be the utility function of user $i$ when joining a coalition $G$, which maps $G$ to a numerical benefit that $i$ will perceive. For clarity, we set the standalone utility $u_i(\{i\}) = 0$ when $i$ is alone.

Given a coalition structure ${\mathcal P} \in {\mathscr P}^K$, a coalition of participants $G$ of size at most $K$ is called a {\em blocking coalition} with respect to ${\mathcal P}$, if $G \notin {\mathcal P}$, and all participants in $G$ can {\it strictly} increase their utilities when they form a coalition $G$ instead of any coalition $G'$ in ${\mathcal P}$, namely,
\begin{equation}
u_i(G) > u_i(G') \mbox{\ for all\ }  i \in G, G' \in {\mathcal P} \mbox{\ where\ }  i \in G'
\end{equation}

A coalition structure is called a {\em stable coalition structure}, denoted by $\hat{\mathcal P} \in {\mathscr P}^K$, if there exists no blocking coalition with respect to $\hat{\mathcal P}$. Note that a stable coalition structure is also a strong Nash equilibrium\footnote{A strong Nash equilibrium is a Nash equilibrium, in which no group of participants can cooperatively deviate in an allowable way that strictly benefits all of its members.}.  In a stable coalition structure, the utility for every participant is always non-negative $u_i(G) \ge 0$. Otherwise, the participant will not join any coalition because of $u_i(\{i\}) = 0$. Note that there may or may not exist a stable coalition structure in a hedonic game. In general, determining the existence of a stable coalition structure is NP-hard \cite{AB12,H13} (specifically, $\Sigma_2^p$-complete \cite{W13}).

\subsection{Cost-Sharing Mechanisms}

In this paper, we specifically consider coalition formation for cost-sharing, and the induced hedonic games by certain cost-sharing mechanisms. In such a setting, there is a non-negative cost function of each coalition $G$, denoted by $C(G)$. Also, denote by $C_i \triangleq C(\{ i \})$ the {\em standalone} (or default) cost for participant $i$, when $i$ is alone.

We note an important property called {\em cost monotonicity}, which holds intuitively in many cost-sharing applications: 
\begin{equation}
C(H) \le C(G), \quad \mbox{ if } H \subseteq G.
\end{equation}
Namely, a larger coalition should incur a larger cost. 

By agreeing to form a coalition $G$, the participants in $G$ are supposed to share the cost $C(G)$. A {\em cost-sharing mechanism} is characterized by a payment function $p_{i}(G)$, which is the shared cost of participant $i \in G$. 
A cost-sharing mechanism $p_{i}(\cdot)$ is said to be {\em budget balanced}, if $\sum_{i \in G} p_{i}(G) = C(G)$ for every $G \subseteq {\mathcal N}$. 

Given a coalition $G$, the utility function of participant $i \in G$ can be derived from the payment function as follows:
\begin{equation} \label{eqn:uipi}
u_i(G) = C_i - p_i(G)
\end{equation}
Namely,  the utility function measures the surplus of joining coalition $G$, as compared with being alone.

In this paper, we consider the following simple well-defined cost-sharing mechanisms (denoted by different superscripts):
\begin{enumerate}

\item {\bf Equal-split Cost-Sharing}: The cost is split equally among all participants: 
\begin{equation}
p^{\rm eq}_{i}(G) \triangleq \frac{C(G)}{|G|}
\end{equation}
Namely, $u^{\rm eq}_i(G) = C_i  - \frac{C(G)}{|G|}$. 

\item {\bf Proportional-split Cost-Sharing}: The cost is split proportionally according to the participants' standalone costs: 
\begin{equation}
p^{\rm pp}_{i}(G)\triangleq \frac{C_i \cdot C(G)}{\sum_{j \in G} C_j} 	
\end{equation}
Namely, $u^{\rm pp}_i(G) = C_i \cdot \frac{(\sum_{j \in G} C_j) - C(G)}{\sum_{j \in G} C_j}$. 

\item {\bf Bargaining-based Cost-Sharing}: The cost-sharing problem can be formulated as a bargaining game \cite{bargainbook}, with a feasible set, and a disagreement point, to which the participants will fall back on when no coalition is formed. Under the bargaining game model, the feasible set is the set of utility values $(\hat{u}_i)_{i \in G}$, such that $\sum_{i \in G} \hat{u}_i \le \sum_{i \in G}u_i(G)$ (equivalently, $\sum_{i \in G} p_i \ge C(G)$), and the disagreement point is $(\hat{u}_i = 0)_{i \in G}$ when each participant pays only the respective standalone cost.  There are two common bargaining solutions in the literature \cite{bargainbook}:
\begin{enumerate}

\item {\bf Egalitarian-split Cost-Sharing} is given by:
\begin{equation}
p^{\rm ega}_{i}(G) \triangleq C_i - \frac{(\sum_{j \in G} C_j) - C(G)}{|G|}
\end{equation}
Namely, all participants $i \in G$ will receive the same utility as $u^{\rm ega}_i(G) = \frac{(\sum_{j \in G} C_j) - C(G)}{|G|}$.

\item {\bf Nash Bargaining Solution} is given by:
\begin{equation}
\big( p^{\rm ns}_{i}(G) \big)_{i \in G} \in \arg\max_{(\hat{p}_{i})_{i \in G}} \prod_{i \in G} u_i(\hat{p})
\end{equation}
subject to 
\begin{equation}
\qquad  \sum_{i \in G} \hat{p}_i = C(G) \notag
\end{equation}


\end{enumerate}

It can be shown that $p^{\rm ega}_{i}(G) = p^{\rm ns}_{i}(G)$, and hence both egalitarian-split cost-sharing and Nash bargaining solution produce equivalent coalition structures\footnote{Furthermore, in egalitarian-split and Nash bargaining cost-sharing, non-positive payment (i.e., $p^{\rm ega}_{i}(G) < 0$ or $p^{\rm ns}_{i}(G) < 0$) is possible because it may need to compensate those with low standalone costs to reach equal utility among all participants. \cite{CE17sharing} also showed the equivalence of stable coalition structures in egalitarian-split and Nash bargaining cost-sharing with and without non-negative payment constraints.} \cite{CE17sharing}. 

\end{enumerate}

\begin{figure}[t!] 
\centering
	\includegraphics[width=0.35\textwidth]{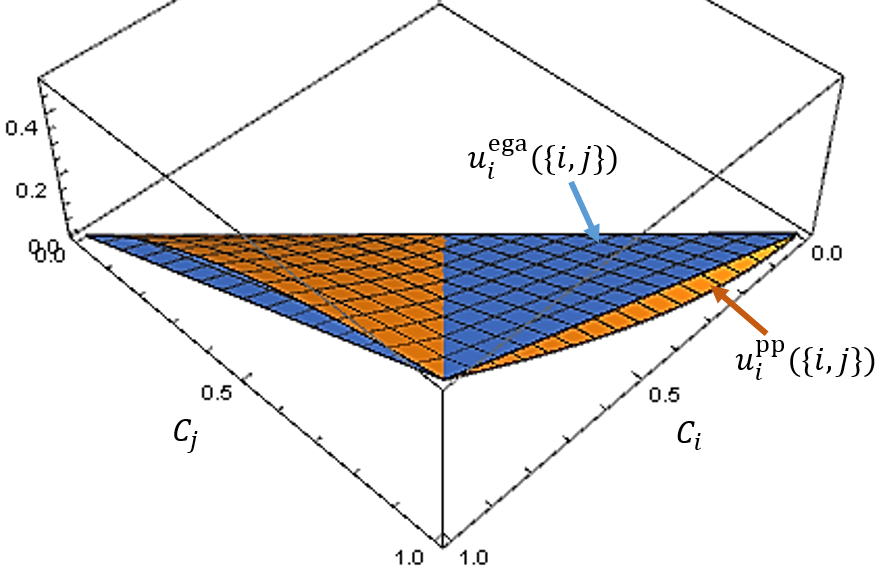}  
	\caption{Comparison of $u^{\rm ega}_{i}(\{i,j\})$ (blue) and $u^{\rm pp}_{i}(\{i,j\})$  (orange). If $C_j \ge C_i$, $u^{\rm ega}_{i}(\{i,j\}) \ge u^{\rm pp}_{i}(\{i,j\})$. Otherwise, $u^{\rm pp}_{i}(\{i,j\}) \ge u^{\rm eqa}_{i}(\{i,j\})$. }
	\label{fig:compare}
\end{figure}

{\bf Remarks:}
While the equal-split cost-sharing mechanism distributes the cost equally to every user regardless of their standalone costs, proportional-split and egalitarian-split cost-sharing mechanisms distribute the cost with varying degrees. For example, $G = \{i,j\}$, we plot $u^{\rm ega}_{i}(\{i,j\})$ and $u^{\rm pp}_{i}(\{i,j\})$ according to $C_i$ and $C_j$ in Fig.~\ref{fig:compare}, assuming $C(\{i,j\})=1$. If $C_j \ge C_i$, then $u^{\rm ega}_{i}(\{i,j\}) \ge u^{\rm pp}_{i}(\{i,j\})$. Otherwise, $u^{\rm pp}_{i}(\{i,j\}) \ge u^{\rm eqa}_{i}(\{i,j\})$. Namely, egalitarian-split cost-sharing favors smaller standalone costs, whereas proportional-split cost-sharing favors larger ones.

Unlike general hedonic games, certain cost-sharing mechanisms can guarantee the existence of a stable coalition structure as given by the following proposition.

\begin{proposition}[\cite{CE17sharing}] \label{prop:exist}
There exist stable coalition structures under equal-split, proportional-split and egalitarian-split cost-sharing mechanisms, and Nash bargaining solution.
\end{proposition}

Another cost-sharing mechanism called usage-based cost-sharing mechanism is also considered in \cite{CE17sharing}, which does not always guarantee the existence of a stable coalition structure.

\subsection{Strong Price-of-Anarchy}

Given a coalition structure ${\mathcal P}$, let  $u(G) \triangleq \sum_{i \in G} u_i(G)$ and $u({\mathcal P}) \triangleq \sum_{ G \in {\mathcal P}} u(G)$. We call $u({\mathcal P})$ the social utility of ${\mathcal P}$.
A {\em social optimum} is a coalition structure that maximizes the total social utility of all users: ${\mathcal P}^\ast = \arg\max_{{\mathcal P} \in {\mathscr P}_K} u({\mathcal P})$. 

We define the utility-based {\em Strong Price of Anarchy} (SPoA) as the worst-case ratio between the social utility of a stable coalition structure and that of a social optimum over any instance of $u(\cdot)$ and its stable coalition structure:
\begin{equation}
{\sf SPoA}_K^u \triangleq \max_{u(\cdot), \hat{\mathcal P}}\frac{u({\mathcal P}^\ast)}{u(\hat{\mathcal P})}
\end{equation}
Specifically, the SPoA when using specific cost-sharing mechanisms is  denoted by ${\sf SPoA}^{u,{\rm eq}}_K$,  ${\sf SPoA}^{u,{\rm pp}}_K$, ${\sf SPoA}^{u,{\rm ega}}_K$, respectively. SPoA provides a metric of measuring social optimality.

We define the social cost of ${\mathcal P}$ by $C({\mathcal P}) \triangleq \sum_{ G \in {\mathcal P}} C(G)$. For any budget balanced cost-sharing mechanism, $C({\mathcal P}) = \sum_{i \in G, G \in {\mathcal P}}C_i-u({\mathcal P})$ (which follows from Eqn.~\raf{eqn:uipi}). Hence,  a social optimum equivalently minimizes the social cost of all users: ${\mathcal P}^\ast = \arg\min_{{\mathcal P} \in {\mathscr P}_K} C({\mathcal P})$. 

For cost-sharing applications, we define the cost-based SPoA with respect to cost over any instance of $C(\cdot)$ and its stable coalition structure:
\begin{equation}
{\sf SPoA}_K^C \triangleq \max_{C(\cdot), \hat{\mathcal P}}\frac{C(\hat{\mathcal P})}{C({\mathcal P}^\ast)}
\end{equation}

For arbitrary hedonic games and cost-sharing mechanisms, the SPoA can be as large as $K$:  ${\sf SPoA}_u^C = O(K)$ and ${\sf SPoA}_K^C = O(K)$. Next, we will see that certain cost-sharing mechanisms can induce only logarithmic SPoA.

\subsection{Mixed Cost-Sharing Mechanisms}

Previously, we only considered one single cost-sharing mechanism in a stable coalition structure. In this section, we extend a stable coalition structure to consider multiple cost-sharing mechanisms.
Equal-split, proportional-split and egalitarian-split cost-sharing mechanisms are called {\em pure} cost-sharing mechanisms. We then consider a mixed setting with these pure cost-sharing mechanisms. A {\em mixed} cost-sharing mechanism consists of a set of constituent pure cost-sharing mechanisms. For a coalition structure ${\mathcal P}$, the cost of $C(G)$, where $G \in {\mathcal P}$, is divided according to one of the constituent pure cost-sharing mechanisms for every participant $i \in G$. A stable coalition structure with respect to a mixed cost-sharing mechanism exists, if there is no blocking coalition induced by any constituent pure cost-sharing mechanism. Note that even if its constituent pure cost-sharing mechanisms always guarantee the existence of stable coalition structures, a mixed cost-sharing mechanism may not guarantee the existence of a stable coalition structure.

\section{Application to P2P Energy Sharing} \label{sec:p2penergy}

In this section, we apply decentralized coalition formation to the scenario of P2P energy sharing.  
We follow the model of P2P energy sharing in our prior work \cite{chau19p2penergy, chau20coalition}. Without sharing, each user will only use his local energy resources (e.g., rooftop PV and home battery) to satisfy local energy demand, in addition to acquiring energy from the grid. However, local energy resources may be underutilized or out of capacity. By sharing mutual energy resources in P2P energy sharing, the users can optimize the efficiency of their energy resources. In practice, energy importing and exporting among users can be achieved by ``virtual net metering'', which is a billing process at the grid operator that allows the credits of energy export of one user to offset the debits of energy import of another.

 \begin{figure}[t!] 
        \centering \includegraphics[width=0.35\textwidth]{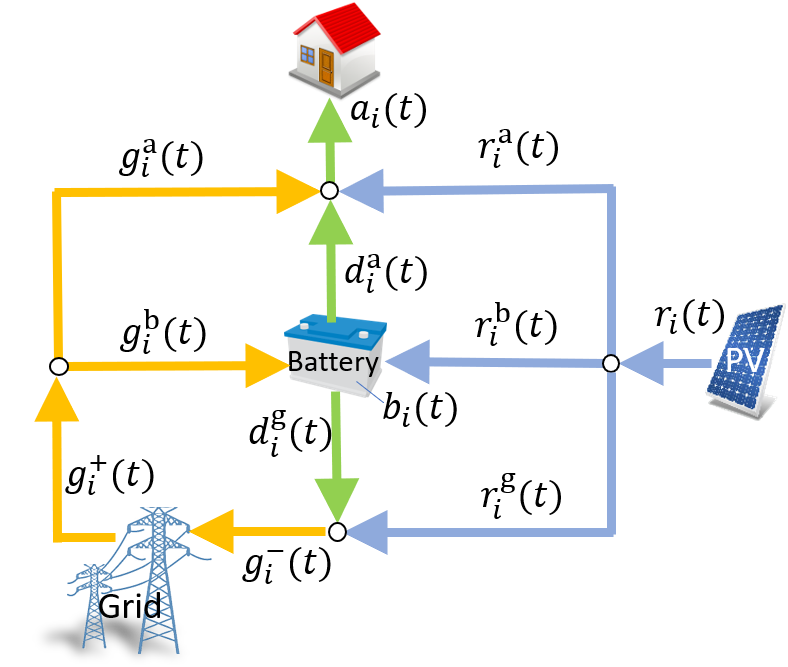}  
        \caption{Variables of the energy management system.}
        \label{fig:system}
    \end{figure}
	
We consider the scenario of allowing at most $K$ users to establish a P2P energy sharing agreement in a coalition for an extended period of time. We limit the size $K$, as to reduce the control overhead of multiple energy resources.
We first define the variables (Table~\ref{tab:notations}) of a user's energy management system, whose relationships are illustrated in Fig.~\ref{fig:system}:

\begin{itemize}
\item
{\bf Demand:}  Each user $i \in {\mathcal N}$ is characterized by an energy demand function $a_i(t)$. The demand function can be estimated based on the prediction from historical data. 

\item
{\bf Rooftop PV:}  Each user $i \in {\mathcal N}$ is equipped with rooftop PV, characterized by an energy supply function $r_i(t)$. The supply function $r_i(t)$ is divided into three feed-in rates: ${r}^{\rm a}_i(t)$ is for local demand, ${r}^{\rm b}_i(t)$ for charging battery, and ${r}^{\rm g}_i(t)$ for electricity feed-in to the grid. 

\item
{\bf Home Battery:}  Each user $i \in {\mathcal N}$ is equipped with a home battery, characterized by capacity ${\sf B}_i$. The battery is constrained by charging efficiency $\eta_{\rm c} \le 1$ and discharging efficiency $\eta_{\rm d} \ge 1$, charge rate constraint $\mu_{\rm c}$ and discharge rate constraint $\mu_{\rm d}$. Let $b_i(t) $ be the current state-of-charge in the battery at time $t$, and ${d}^{\rm a}_i(t)$ be the discharge rate for local demand, whereas ${d}^{\rm g}_i(t)$ be the discharge rate for electricity feed-in to the grid. 

\item
{\bf Grid:} Each user can also import electricity from the grid if his demand is not entirely satisfied. Further, rooftop PV and home battery can inject excessive electricity into the grid, when feed-in compensation is offered by a grid operator. Let ${\sf C}_{\rm g}^+$ be the per-unit cost by the grid on electricity consumption, and  ${\sf C}_{\rm g}^-$ be the per-unit compensation on electricity feed-in. Let $g_i^+(t)$ be the electricity consumption rate of the user $i$ at time $t$, and $g_i^-(t)$ be the electricity feed-in rate. Let ${g}^{\rm a}_i(t)$ be the consumption rate for local demand, and ${g}^{\rm b}_i(t)$ be the consumption rate for charging a battery. 

\end{itemize}

\begin{table}[t!]  
\begin{tabular}{@{}c@{}|p{0.85\columnwidth}@{}}
\hline \hline
${\mathcal N}$  & The set of users   \\
$a_i(t)$  &  Energy demand function of user $i$ \\
$r_i(t)$  &  Energy supply function of user $i$ \\
${r}^{\rm a}_i(t)$ & Energy supply for local demand of user $i$\\
${r}^{\rm b}_i(t)$ & Energy supply for battery charging of user $i$\\
${r}^{\rm g}_i(t)$ & Energy supply for feed-in to the grid of user $i$\\
${\sf B}_i$ & Battery capacity of user $i$\\
$\eta_{\rm c},\eta_{\rm d}$ & Charging and discharging efficiency \\
$\mu_{\rm c},\mu_{\rm d}$ & Charge and discharge rate constraints \\
$b_i(t)$ & Current state-of-charge in the battery at time $t$ \\
${d}^{\rm a}_i(t)$ & Discharge rate for local demand of user $i$\\
${d}^{\rm g}_i(t)$ & Discharge rate for feed-in to the grid of user $i$\\
${\sf C}_{\rm g}^+$ & Per-unit cost by the grid on electricity consumption\\
${\sf C}_{\rm g}^-$ & Per-unit compensation on electricity feed-in\\
$g_i^+(t)$ & Electricity consumption rate of user $i$ at time $t$\\
$g_i^-(t)$ &  Electricity feed-in rate of user $i$	\\
${g}^{\rm a}_i(t)$ & Consumption rate for local demand of user $i$\\ 
${g}^{\rm b}_i(t)$ & Consumption rate for battery charging of user $i$\\
\hline \hline
\end{tabular}
\caption{Key notations of variables.} \label{tab:notations}
\end{table}


If a group of users $G \subseteq {\mathcal N}$ form a coalition to share their energy resources and to minimize their total operational cost, then ${C}(G)$ can be defined as the minimum cost of the energy management optimization. We formulate the energy management optimization problem in {\sf (EP)}, with the constraints of variables captured in Fig.~\ref{fig:system}.

\begin{figure}[!htb]
\begin{align}
 {\sf (EP)}\ & \ {C}(G) \triangleq \min  \sum_{t=1}^{T}  \sum_{i \in G}  \big({\sf C}_{\rm g}^+ {g}^+_i(t)- {\sf C}_{\rm g}^-  g^-_i(t) + {\sf C}_{\rm s} {s}^+_i(t) \big)  \notag\\
\textrm{s.t.\ } &   b_i(t+1) - b_i(t)  =  \eta_{\rm c} ( {r}^{\rm b}_i(t) + {g}^{\rm b}_i(t) ) \notag \\
 & \qquad \qquad \qquad \quad  - \eta_{\rm d} ({d}^{\rm a}_i(t) + {d}^{\rm g}_i(t) ), & \\
 &  0 \le b_i(t) \le {\sf B}_i, b_i(0)=0,  &  \label{eqn:over-}\\
 &   {g}^{\rm b}_i(t) + {r}^{\rm b}_i(t) \le \mu_{\rm c},  &  \\ 
 &   {d}^{\rm a}_i(t) + {d}^{\rm g}_i(t) \le \mu_{\rm d},  &  \\ 
 &  {d}^{\rm a}_i(t) +  {g}^{\rm a}_i(t) + {r}^{\rm a}_i(t)  = a_i(t),  &  \\
 &   {r}^{\rm a}_i(t) + {r}^{\rm b}_i(t) + {r}^{\rm g}_i(t)  = r_i(t),  &  \\ 
 &   {g}^{\rm a}_i(t) + {g}^{\rm b}_i(t) = g^+_i(t) + {s}^+_i(t),  &  \\ 
 &   {d}^{\rm g}_i(t) + {r}^{\rm g}_i(t) = g^-_i(t) + {s}^-_i(t),  &  \\  
 & \sum_{i \in G} {s}^+_i(t) = \sum_{i \in G} {s}^-_i(t)  & \label{eqn:grid-balance}  \\
 \textrm{var. \ } & b_i(t) \ge 0, {d}^{\rm a}_i(t)\ge 0, {d}^{\rm g}_i(t)\ge 0,  \forall t \in [1,T], \forall i \in G \notag \\
 &  {r}^{\rm a}_i(t)\ge 0, {r}^{\rm b}_i(t) \ge, {r}^{\rm g}_i(t) \ge  0,  \forall t \in [1,T], \forall i \in G &\notag \\
 &  {g}^{\rm a}_i(t)\ge 0, {g}^{\rm b}_i(t) \ge 0, {s}^+_i(t) \ge 0, {s}^-_i(t) \ge 0 &    \notag 
\end{align}
\end{figure}

In {\sf (EP)}, there is a per-unit service fee ${\sf C}_{\rm s}$ charged by the grid operator for virtual net metering. Let ${s}^+_i(t)$ and ${s}^-_i(t)$ be the consumption rate and the feed-in rate under virtual net metering for user $i$, respectively. We assume ${\sf C}_{\rm s} < {\sf C}_{\rm g}^-$. Otherwise, there is no need for energy sharing.

{\bf Remarks}: The Service fee ${\sf C}_{\rm s}$ is an important factor to the viability of coalition-based optimization. A high service fee ${\sf C}_{\rm s}$ will deter coalition formation in P2P energy sharing. As long as ${\sf C}_{\rm s}  < {\sf C}_{\rm g}^+  - {\sf C}_{\rm g}^- $, it is viable to transfer PV energy via virtual net metering, because the price difference ${\sf C}_{\rm g}^+  - {\sf C}_{\rm g}^-$ (that represents simultaneous importing and exporting electricity) is still higher than virtual net metering at the cost ${\sf C}_{\rm s}$. Even though ${\sf C}_{\rm s} >  {\sf C}_{\rm g}^-$, when ${\sf C}_{\rm g}^+$ is sufficiently large, we still have ${\sf C}_{\rm s}  < {\sf C}_{\rm g}^+  - {\sf C}_{\rm g}^-$. We assume that all users are energy consumers rather than energy producers. Hence, the standalone cost $C_i$ is always positive for all users.

\section{Outline of Results}

We present the theoretical results of the general coalition formation model. The full proofs can be found in the Appendix.

The cost-based SPoA has been studied in our prior work \cite{CE17sharing}. We can relate utility-based SPoA to cost-based SPoA by Theorem~\ref{thm:poau}.

\begin{customthm}{1} \label{thm:poau}
Consider a budget balanced cost-sharing mechanism $p_{i}(\cdot)$. For the cost-based and utility-based strong prices of anarchy ${\sf SPoA}_K^C$ and ${\sf SPoA}_K^u$, respectively, we have
\begin{equation}
{\sf SPoA}_K^u \ge \frac{K-1}{K-{\sf SPoA}_K^C}; \quad
{\sf SPoA}_K^C \le K-\frac{K-1}{{\sf SPoA}_K^u}
\end{equation}
\end{customthm}

\medskip

There are some known results on the cost-based SPoA: 

\begin{proposition}[\cite{CE17sharing}] \label{prop:eqSPoA}
The SPoA of equal-split and proportional-split cost-sharing mechanisms with cost monotonicity are both logarithmically bounded: ${\sf SPoA}^{C,{\rm eq}}_K = O(\log K)$ and ${\sf SPoA}^{C,{\rm pp}}_K = O(\log K)$.
\end{proposition}

\begin{proposition}[\cite{CE17sharing}] \label{prop:egaSPoA}
The SPoA of egalitarian-split cost-sharing mechanism and Nash bargaining
solution with cost monotonicity is upper bounded by ${\sf SPoA}^{C,{\rm ega}}_K = {\sf SPoA}^{C,{\rm ns}}_K = O(\sqrt{K}\log K)$.
\end{proposition}

In this paper, we are able to improve the bound of SPoA of egalitarian-split cost-sharing mechanism by Theorem~\ref{thm:egaSPoA}.

\begin{customthm}{2}  \label{thm:egaSPoA}
For egalitarian-split cost-sharing mechanism and Nash bargaining solution with cost monotonicity, the SPoA is upper bounded by
\begin{equation}
{\sf SPoA}^{C,{\rm ega}}_K  ={\sf SPoA}^{C,\rm ns}_K= O(\log K)
\end{equation} 
\end{customthm}

{\bf Remarks:} Hence, we can conclude that equal-split, proportional-split and egalitarian-split cost-sharing mechanisms and Nash bargaining solution share the same logarithmic order of magnitude in SPoA. Let $\hat{\mathcal P}^{\rm eq}, \hat{\mathcal P}^{\rm pp}, \hat{\mathcal P}^{\rm ega}$ be the stable coalition structures for equal-split, proportional-split and egalitarian-split cost-sharing mechanisms, respectively. Because of the benchmark against the social optimum in Propositions~\ref{prop:eqSPoA}-\ref{prop:egaSPoA} and Theorem~\ref{thm:egaSPoA}, we can also bound the worst-case ratio of stable coalition structures among different cost-sharing mechanisms as follows:  
\begin{align}
\max_{\hat{\mathcal P}^{\rm eq}, \hat{\mathcal P}^{\rm pp}, C(\cdot)}\max\Big\{ \frac{C(\hat{\mathcal P}^{\rm pp})}{C(\hat{\mathcal P}^{\rm eq})}, \frac{C(\hat{\mathcal P}^{\rm eq})}{C(\hat{\mathcal P}^{\rm pp})} \Big\} = O(\log K) \notag \\
\max_{\hat{\mathcal P}^{\rm pp}, \hat{\mathcal P}^{\rm ega}, C(\cdot)}\max\Big\{\frac{C(\hat{\mathcal P}^{\rm pp})}{C(\hat{\mathcal P}^{\rm ega})},\frac{C(\hat{\mathcal P}^{\rm ega})}{C(\hat{\mathcal P}^{\rm pp})} \Big\} = O(\log K) \notag \\
\max_{\hat{\mathcal P}^{\rm ega}, \hat{\mathcal P}^{\rm eq}, C(\cdot)}\max\Big\{\frac{C(\hat{\mathcal P}^{\rm ega})}{C(\hat{\mathcal P}^{\rm eq})},\frac{C(\hat{\mathcal P}^{\rm eq})}{C(\hat{\mathcal P}^{\rm ega})} \Big\} = O(\log K) \notag 
\end{align} 
Namely, the costs of these cost-sharing mechanisms are not far from each other.

\medskip

We next provide a logarithmic lower bound on SPoA for practical polynomial-time mechanisms in Theorem~\ref{thm:lowerbound}. Hence, it shows that equal-split, proportional-split and egalitarian-split cost-sharing mechanisms and Nash bargaining solution are the best practical mechanisms with minimal SPoA.

\begin{customthm}{3}  \label{thm:lowerbound}
For any cost-sharing mechanism, such that a stable coalition structure can be found in polynomial-time, its SPoA is lower bounded by $\Omega(\log K)$, under the assumption that P$\neq$NP.
\end{customthm}

\medskip

We can also derive the SPoA for a mixed cost-sharing mechanism as the the upper bound of the SPoA of its constituent pure cost-sharing mechanisms in Theorem~\ref{thm:mixed}.

\begin{customthm}{4}  \label{thm:mixed}
For a mixed cost-sharing mechanism of a constant number of constituent pure cost-sharing mechanisms, if the SPoA of each constituent pure cost-sharing mechanism is $O(f(K))$, then the SPoA of the mixed cost-sharing mechanism is also $O(f(K))$.
\end{customthm}

\section{Decentralized Coalition Formation Algorithm} 

In this section, we present a decentralized coalition formation algorithm. In our previous paper \cite{CE17sharing}, we have presented a centralized coalition formation algorithm, which is based on sequential steps of pruning the transition graph of states, which can be very slow in practice if there are many states in coalition formation. This decentralized coalition formation process is based on the classical deferred-acceptance algorithms (e.g., the Gale-Shapley algorithm for the stable marriage problem and the Irving algorithm for the stable roommates problem). We extend the deferred-acceptance algorithm to coalition formation with more than 2 participants per coalition. The deferred-acceptance algorithm allows parallelization, such that multiple users can propose coalitions to each other simultaneously. 

We first define some notations. For each participant $i \in {\mathcal N}$, let ${\mathcal G}_i \triangleq \{G  \subseteq {\mathcal N} :  i \in G \mbox{\ and\ } G \in {\mathcal P} \in {\mathscr P}^K\}$ be the set of {\it feasible} coalitions that include participant $i$. Note that ${\mathcal G}_i$ also includes the standalone coalition $\{i\}$. Each participant $i$ has a preference over ${\mathcal G}_i$, defined by $u_i(\cdot)$. Denote participant $i$'s preference by an ordered sequence $\mbox{\sc Pref}_i = (G_1, G_2, .., G_t, ..., \{i\})$ such that $u_i(G_{t-1}) \ge u_i(G_{t})$.  For brevity, we only consider the coalition $G$ in $i$'s preference where $u_i(G) > 0 = u_i(\{i\})$. 
We define two operations that enumerate the preference $\mbox{\sc Pref}_i$: 
\begin{enumerate}

\item $\mbox{\sc Top}(\mbox{\sc Pref}_i)$ returns the topmost preferred coalition from $\mbox{\sc Pref}_i$.

\item $\mbox{\sc Remove}(\mbox{\sc Pref}_i, G)$ removes $G$ from $\mbox{\sc Pref}_i$.

\end{enumerate}

During the process of decentralized coalition formation, each participant $i$ carries a tuple of variables:
$$\big\langle \mbox{\sc Props}_i, {\sf H}_i, \mbox{\sc Suspend}_i \big\rangle$$
\begin{itemize}

\item
 $\mbox{\sc Props}_i \subseteq {\mathcal G}_i$ is a set of proposed coalitions received by $i$. For each $G \in \mbox{\sc Props}_i$, we define a function $\mbox{\sc Proposer}(G) \mapsto G$ that indicates its proposer.

\item
 ${\sf H}_i\in {\mathcal G}_i$ is a coalition that is currently held for consideration by $i$.
 
 \item
 $\mbox{\sc Suspend}_i$ is a boolean variable, indicating if $i$ is suspended from proposing the next-top preferred coalition to other participants.
 
 \end{itemize}

Initially, set ${\sf H}_i \leftarrow  \varnothing$ and $\mbox{\sc Suspend}_i \leftarrow  \mbox{\sc False}$, $\mbox{\sc Props}_i \leftarrow \varnothing$. We define a decentralized process, {\sc Coln{-}Form}, consisting of multiple rounds with three stages per each round as follows:
\begin{enumerate}

\item {\bf Proposing Stage:}
First, if participant $i$ is not suspended (i.e., $\mbox{\sc Suspend}_i =\mbox{\sc False}$), then $i$ will propose the next topmost preferred coalition $G = \mbox{\sc Top}(\mbox{\sc Pref}_i)$, when $G$ is better than ${\sf H}_i$, the current one held for consideration (i.e., $u_i(G) > u_i({\sf H}_i)$). Otherwise, $i$ is paused from proposing. Once $i$ has proposed $G$ to all $j \in G\backslash\{i\}$, $G$ is removed from the preference by $\mbox{\sc Remove}(\mbox{\sc Pref}_i, G)$.

\item {\bf Evaluation Stage:}
Then, each participant $j$ collects the set of proposed coalitions $\mbox{\sc Props}_j$ that were received from the proposing stage. If any proposed coalition $G \in \mbox{\sc Props}_j$ is not better than the current one held for consideration (i.e., $u_i({\sf H}_i) > u_i(G)$), then $G$ will be rejected, and $\mbox{\sc Proposer}(G)$ and other members in $G$ will be notified. The rejected coalition $G$ will be removed from the members' proposed coalition sets $\mbox{\sc Props}_k$ for all $k \in G$.

\item {\bf Selection Stage:}
The remaining proposed coalitions in $\mbox{\sc Props}_j$ of each $j$ are not rejected by any members and are better than the ones currently held for consideration. Next, each participant $j$ picks the topmost preferred\footnote{We assume that there is a deterministic tie-breaking for each participants, such that when a pair of coalitions are ranked equally by two participants, they will always carry out tie-breaking in a consistent manner, for example, using the same tie-breaking rule.} coalition $H$ from $\mbox{\sc Props}_j$ and notifies the other members in $H$. If all members $k \in H$ (except $\mbox{\sc Proposer}(H)$) also pick $H$ as their most preferred coalition from the respective lists of proposed coalitions $\mbox{\sc Props}_k$, then $H$ will be held for consideration by all members $H$, and the previous held coalition ${\sf H}_k$ will be replaced by $H$. The members in the previous held coalition $\ell \in {\sf H}_k$ will be notified and their current held coalition will be replaced by $\varnothing$. The previous proposer $\mbox{\sc Proposer}({\sf H}_k)$ will resume to propose in the next round (i.e., $\mbox{\sc Suspend}_i \leftarrow  \mbox{\sc False}$). But the proposer $\mbox{\sc Proposer}(H)$ will be suspended from proposing (i.e., $\mbox{\sc Suspend}_i \leftarrow  \mbox{\sc True}$).

\item {\bf Termination Stage:}
If some participants are not in any coalition held for consideration, then the process proceeds to the next round and repeats the proposing, evaluation and selection stages. Otherwise, the coalition formation process will terminate.

 \end{enumerate}

\subsection{Convergence to Stable Coalition Structure}

We define a {\em cyclic preference} as sequences $(i_1,..., i_t)$ and $(G_1,..., G_t)$, where $i_k \in G_k\cap G_{k+1}$ for all $k \le t-1$, and $i_t \in G_t\cap G_1$, such that
\begin{eqnarray*}
u_{i_1}(G_2) & > & u_{i_1}(G_1), \\
u_{i_2}(G_3) & > & u_{i_2}(G_2), \\
& \vdots & \notag \\
u_{i_t}(G_1) & > & u_{i_t}(G_t)
\end{eqnarray*}

\begin{proposition}[\cite{CE17sharing}] \label{prop:nocyc}
There exists no cyclic preference under equal-split, proportional-split and egalitarian-split cost-sharing mechanisms, and Nash bargaining solution.
\end{proposition}

\begin{customthm}{5}  \label{thm:converge}
If there exists no cyclic preference, then {\sc Coln{-}Form} will converge to a stable coalition structure in time $O(n^K)$.
\end{customthm}

By Proposition~\ref{prop:nocyc} and Theorem~\ref{thm:converge}, {\sc Coln{-}Form} will converge to a stable coalition structure under equal-split, proportional-split and egalitarian-split cost-sharing mechanisms, and Nash bargaining solution. If {\sc Coln{-}Form} is executed sequentially when one participant follows another, then the running time is $O(n^K)$. However, {\sc Coln{-}Form} can also be executed in parallel among the participants, and the actual running time is less than $O(n^K)$.

\begin{figure*}[htb!] 
        \centering\includegraphics[width=0.85\textwidth]{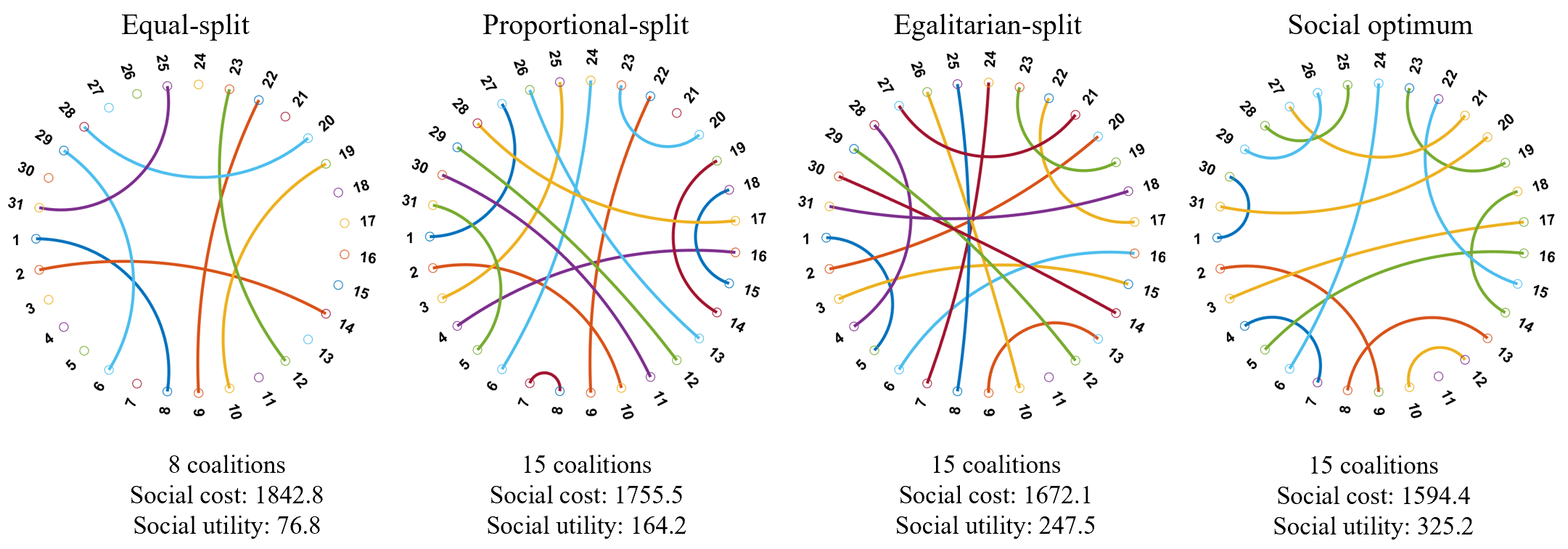}  
        \caption{Coalition structures (when $K=2$) under different cost-sharing mechanisms and social optimum.}
        \label{fig:2coalitiongraph}
 \end{figure*}

\begin{figure*}[htb!] 
        \centering\includegraphics[width=0.85\textwidth]{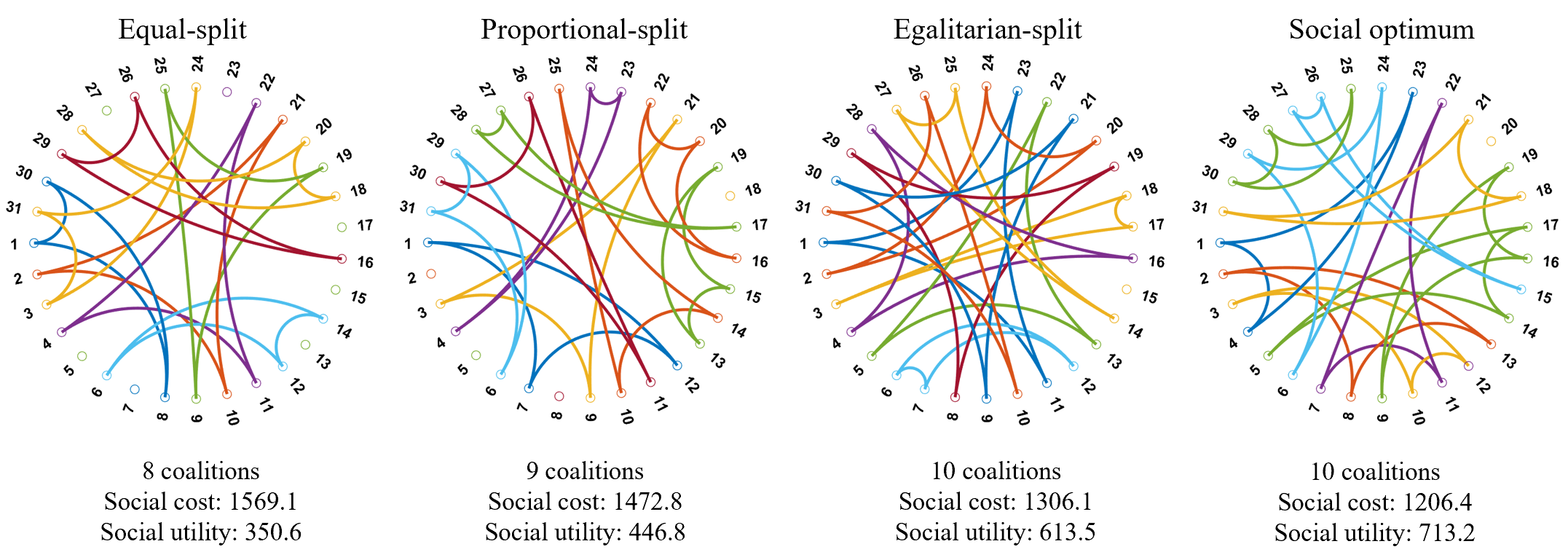}  
        \caption{Coalition structures (when $K=3$) under different cost-sharing mechanisms and social optimum.}
        \label{fig:3coalitiongraph}
 \end{figure*}

\section{Evaluation Study}  \label{sec:eval}

We present an empirical study of decentralized coalition formation in a real-world P2P energy sharing project. A field trial has been conducted on Bruny Island, Tasmania in Australia \cite{SGFJT19bruny}, where approximately 31 batteries were installed with solar PV systems at the homes of selected residents of the island. The batteries include customized management software that allows programmable control of the battery management operations. The participating households provide energy data for empirical studies. Next, we evaluate the outcomes of coalition formation under different cost-sharing mechanisms with 31 households. The evaluation is based on the default parameters in Table~\ref{tab:default_parameters}. 

\begin{table}[htb!]
\centering
	\begin{tabular}{l|r}
		\hline\hline
		Battery capacity (${\sf B}_i,$) & 9.8~kWh \\ \hline
		Consumption tariff (${\sf C}_{\rm g}^+$) & \$0.20/kWh \\ \hline
		Feed-in tariff  (${\sf C}_{\rm g}^-$) & \$0.10/kWh \\ \hline
		Settlement service fee (${\sf C}_{\rm s}$) & \$0.00/kWh \\ \hline
		Charging efficiency ($\eta_{\rm c}$) & 0.95 \\ \hline
		Discharging efficiency ($\eta_{\rm d}$) & 1.05 \\ \hline
		Charge rate ($\mu_{\rm c}$) & 5 \\ \hline
		Discharge rate ($\mu_{\rm d}$) & 5 \\ \hline\hline
	\end{tabular}  
	\caption{Default parameters used in evaluation.} 
	\label{tab:default_parameters} 
\end{table}

The coalition structures for $K=2$ and 3 are visualized in Figs.~\ref{fig:2coalitiongraph}-\ref{fig:3coalitiongraph}. The mechanisms successfully form 8-15 non-singleton coalitions (when $K=2$) and 8-10 non-singleton coalitions (when $K=3$). When $K=2$,  the social optimum has social cost \$1594.4 and social utility \$325.2, whereas the egalitarian-split cost-sharing has social cost \$1672.1 and social utility \$247.5. When $K=3$, the social optimum has social cost \$1206.4 and social utility \$713.2, whereas the egalitarian-split cost-sharing has social cost  \$1306.1 and social utility \$613.5. In both cases, the egalitarian-split cost-sharing is very close to the social optimum. Also, equal-split cost-sharing gives the most social cost and the least social utility among the three cost-sharing  mechanisms.

\begin{figure*}[ht!] 
        \centering\includegraphics[width=0.48\textwidth]{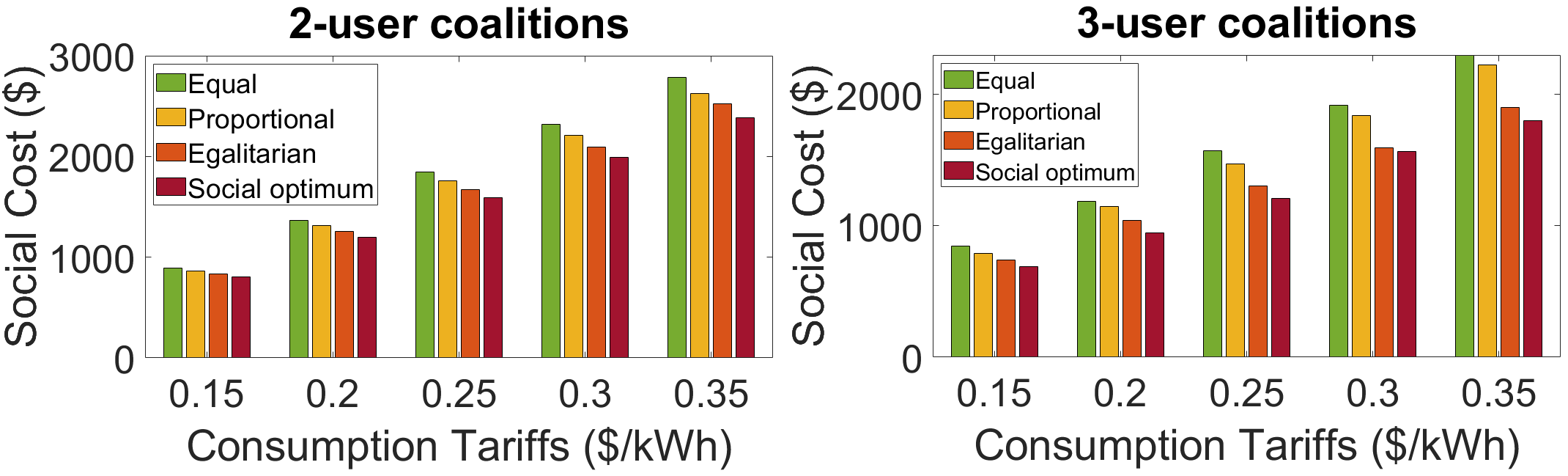} \ \includegraphics[width=0.48\textwidth]{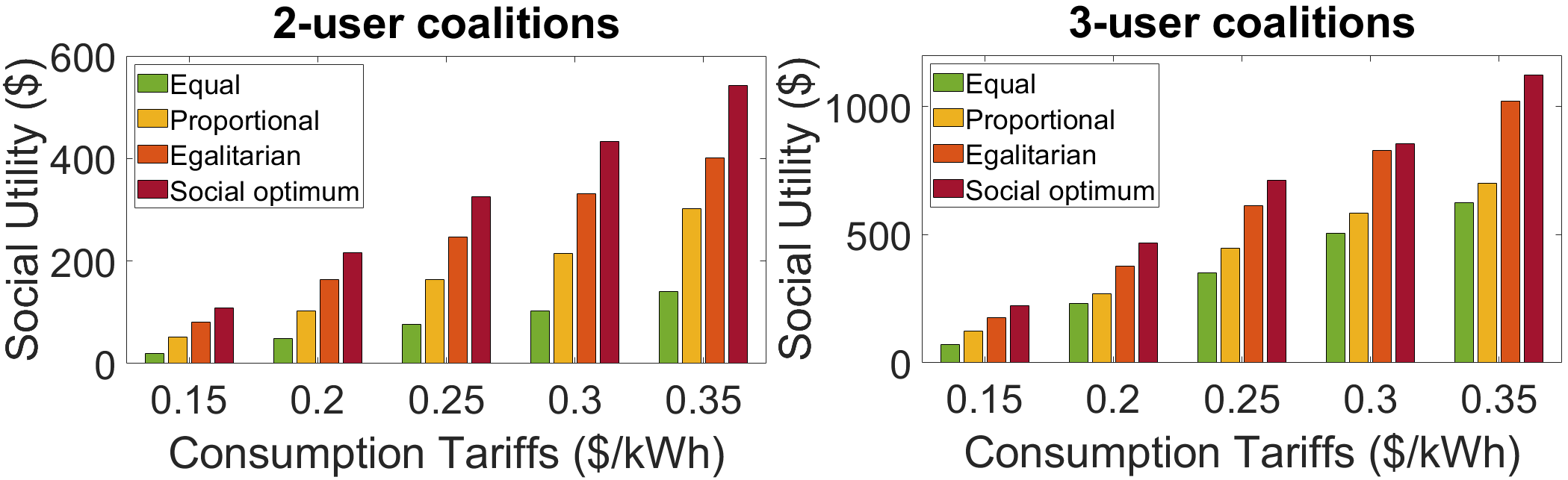} 
        \caption{Evaluation results under different settings of consumption tariffs.}
        \label{fig:eval_consumption}
 \end{figure*}

\begin{figure*}[ht!] 
        \centering\includegraphics[width=0.48\textwidth]{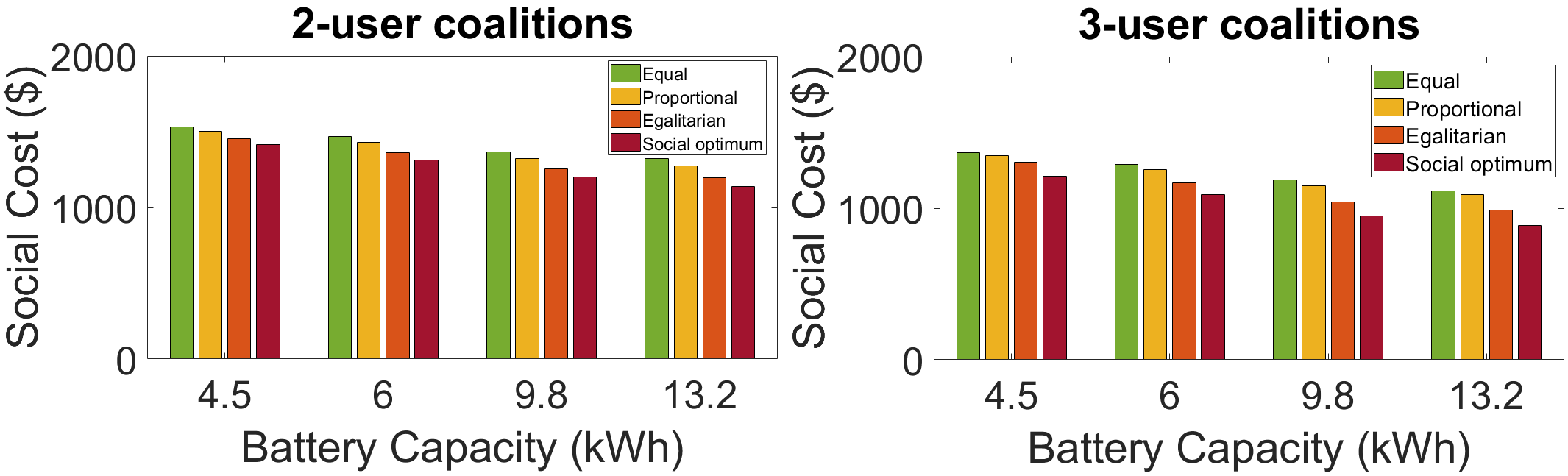} \ \includegraphics[width=0.48\textwidth]{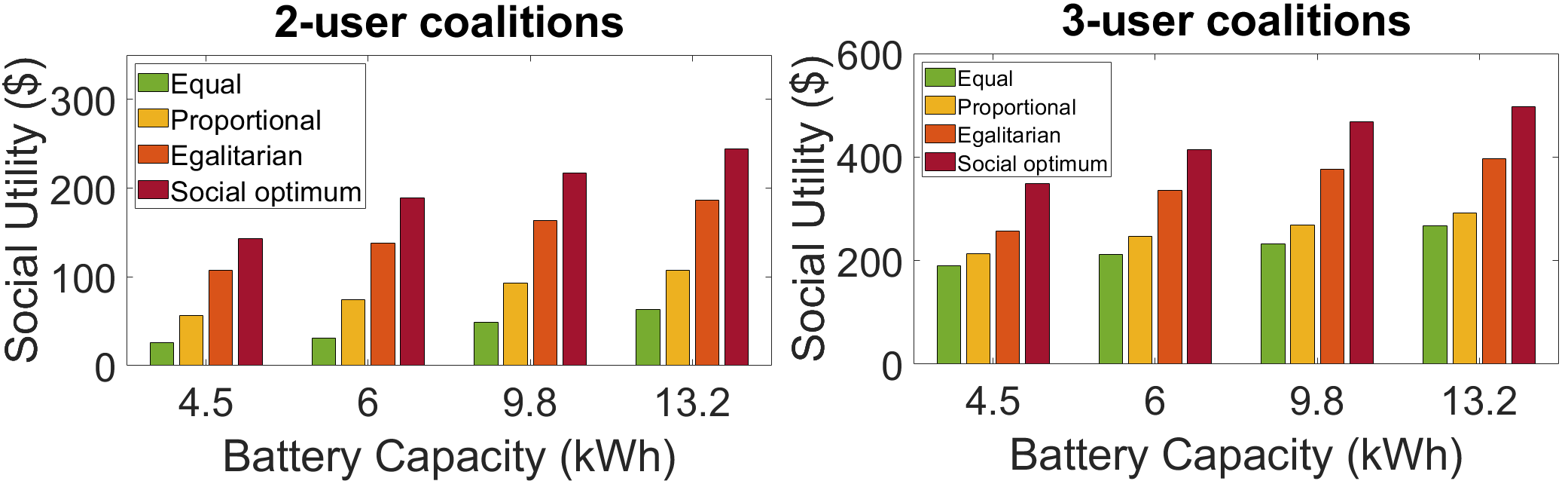} 
        \caption{Evaluation results under different settings of battery capacities.}
        \label{fig:eval_battery}
 \end{figure*}

\begin{figure*}[ht!] 
        \centering\includegraphics[width=0.48\textwidth]{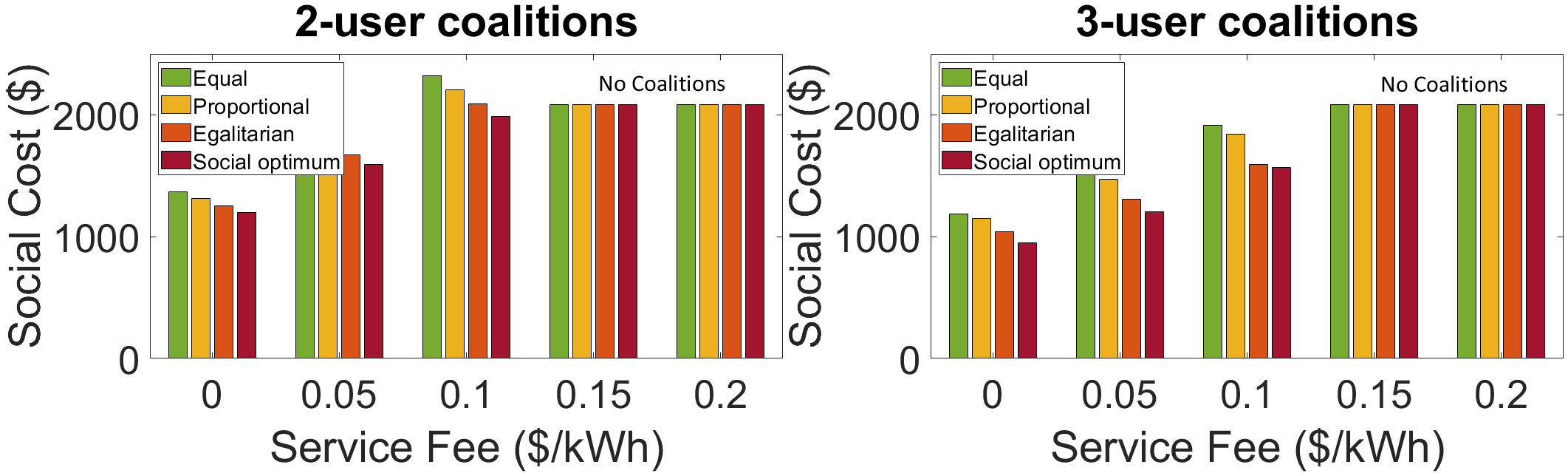} \ \includegraphics[width=0.48\textwidth]{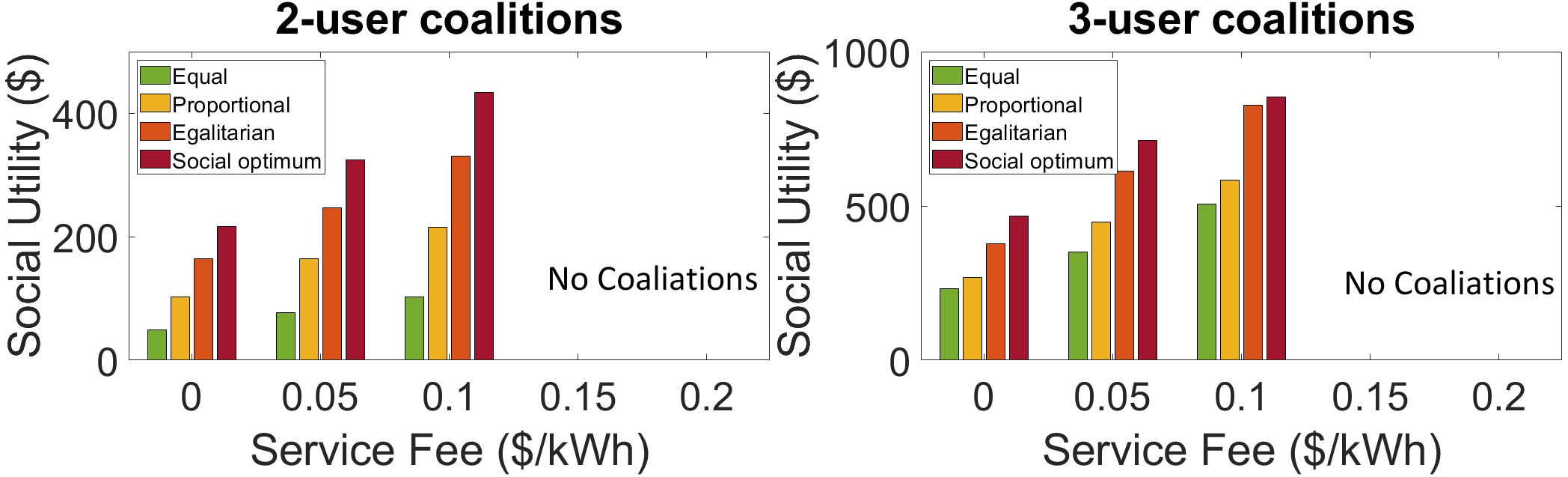}
        \caption{Evaluation results under different settings of service fees.}
        \label{fig:eval_service}
 \end{figure*}

Although the theoretical SPoA are in the same logarithmic order of magnitude for equal-split, proportional-split and egalitarian-split cost-sharing mechanisms, there are observable differences in empirical evaluation. From Figs.~\ref{fig:eval_consumption}-\ref{fig:eval_service}, egalitarian-split cost-sharing is observed to have the closest approximation to the social optimum, whereas equal-split cost-sharing is the furthest from the social optimum,  among the three cost-sharing mechanisms. Overall, we observe the measured SPoA is within $95\%$ of the social optimal cost with coalitions of 2 and 3 users, via three fair cost-sharing mechanisms. Next, we vary the parameters to compare the outcomes of different cost-sharing mechanisms.
 
\subsection{Consumption Tariffs}

The consumption tariff is the cost per kWh of energy imported from the grid. The results in Fig.~\ref{fig:eval_consumption} suggest an increase in the consumption tariff causes an increase in the social cost. Nonetheless, the relative differences among these cost-sharing mechanisms are preserved under different consumption tariffs. 

\subsection{Battery Capacities}

Different battery capacities were considered to assess their impact on decentralized coalition formation. The results in Fig.~\ref{fig:eval_battery} suggest that larger batteries can result in greater social utilities.  The largest and smallest batteries have the capacities of 13.2~kWh and 4.5~kWh, respectively. Despite the larger battery having almost three times the capacity, it has less than twice the utility. This suggests that battery capacity has a lesser impact than consumption tariff. Also, the relative differences among the cost-sharing mechanisms are preserved under different battery capacities.

\subsection{Service Fees}

Different service fees also affect the outcomes of coalition formation. We observe that there is no decentralized coalition formed if the  service fee  ${\sf C}_{\rm s}$ is higher than $ {\sf C}_{\rm g}^+  - {\sf C}_{\rm g}^-$.  The results in Fig.~\ref{fig:eval_service} suggest that when ${\sf C}_{\rm s}$ is greater than $\$0.15$/kWh, there is no decentralized coalition formed -- the users would rather have standalone energy management without sharing energy with other users.

\subsection{Performance of Decentralized Coalition Formation Algorithm}

In this section, we compare the performance of the decentralized algorithm for coalition formation and the centralized algorithm in \cite{CE17sharing}.  We observe that the running time of the decentralized one is much faster in practice in Fig.~\ref{fig:runtime}. The result shows a significant improvement of running time over the centralized algorithm as the number of users increased, on the contrary, the running time of the decentralized algorithm rises moderately. The decentralized algorithm is up to 2.7 times faster than the centralized algorithm with 10 users to around 8000 times faster with 16 users. The running time of the decentralized algorithm has a significant advantage when the number of users reaches more than 12.

 \begin{figure}[h!] 
        \centering \includegraphics[width=0.48\textwidth]{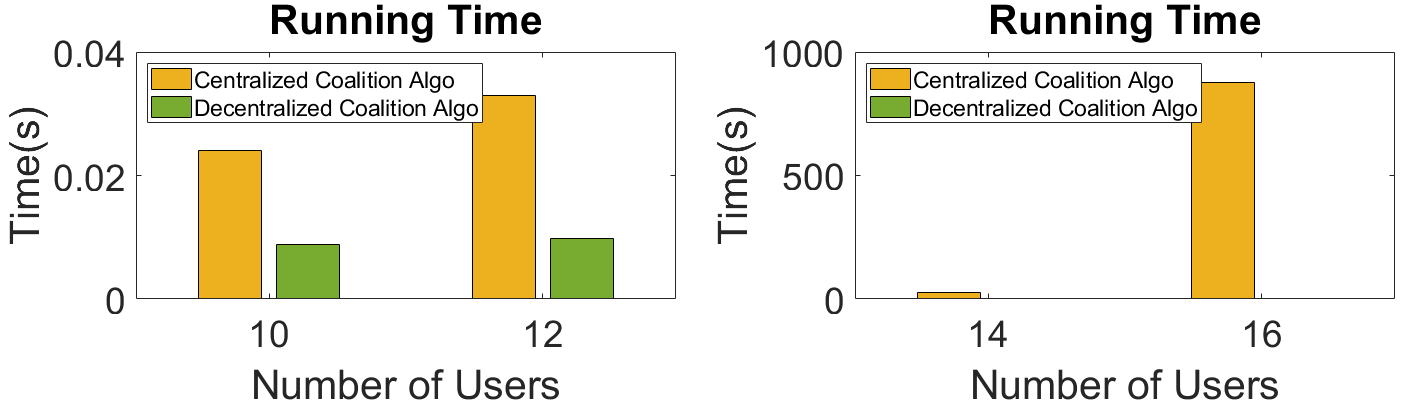}  
        \caption{Compare the running time of centralized and decentralized coalition formation algorithms under different numbers of users.}
        \label{fig:runtime}
    \end{figure}

\section{Conclusion}  \label{sec:conclude}

In this paper, we studied decentralized coalition formation under certain cost-sharing mechanisms (e.g., equal-split, proportional-split, egalitarian and Nash bargaining solutions of bargaining games). We applied our decentralized coalition formation results to the application of P2P energy sharing, being corroborated by an empirical study. We provide several new results of decentralized coalition formation: (1) We established a logarithmic lower bound on SPoA, and hence, showed several previously known fair cost-sharing mechanisms are the best practical mechanisms with minimal SPoA. (2) We improved the SPoA of egalitarian and Nash bargaining cost-sharing mechanisms to match the lower bound. (3) We derived the SPoA of a mix of different cost-sharing mechanisms. (4) We presented a decentralized algorithm to form a stable coalition structure. In the future work, we will implement the decentralized coalition formation in decentralized blockchain framework to enable transparent and verifiable computation of the cost-sharing mechanisms. Note that we recently applied the cost-sharing mechanisms in this work to privacy-preserving energy storage sharing with blockchain \cite{chau21blockchain,chau22blockchain} and decentralized group purchasing for retail energy plans \cite{chau22grouppur, chau21grouppur} . Furthermore, the current energy sharing model does not consider the real-time power flow constraints (e.g., voltage constraints in the grid). It will be considered in a more advanced model that incorporates realistic power flow constraints.

\bibliographystyle{ACM-Reference-Format}
\bibliography{reference,reference2}

\appendix
\section*{Appendix}

\section{Proofs of Theoretical Results} 

\subsection{Preliminaries}

\begin{lemma}[\cite{CE17sharing}]  \label{thm:gen} 
Recall that ${\mathcal P}_{\rm self} \triangleq \big\{\{i\} : i \in {\mathcal N} \big\}$. A social optimum is denoted by ${\mathcal P}^\ast = \arg\max_{{\mathcal P} \in {\mathscr P}_K} u({\mathcal P})$. Then we obtain $K \cdot C({\mathcal P}^\ast) \ge C({\mathcal P}_{\rm self})$.
A stable coalition structure is denoted by $\hat{\mathcal P} \in {\mathscr P}_K$.
Consider a budget balanced cost-sharing mechanism $p_{i}(\cdot)$. Then we obtain
$C({\mathcal P}_{\rm self}) \ge C(\hat{\mathcal P})$.
Hence, the cost-based strong price of anarchy ${\sf SPoA}_K^C \le K$.
\end{lemma}

\begin{customthm}{1}
Consider a budget balanced cost-sharing mechanism $p_{i}(\cdot)$. Denote the cost-based and utility-based strong prices of anarchy by ${\sf SPoA}_K^C$ and ${\sf SPoA}_K^u$, respectively. Then
\begin{equation}\label{e1-1}
{\sf SPoA}_K^u \ge \frac{K-1}{K-{\sf SPoA}_K^C}
\end{equation}
\begin{equation}\label{e1-2}
{\sf SPoA}_K^C \le K-\frac{K-1}{{\sf SPoA}_K^u}
\end{equation}
\end{customthm}

\begin{proof}
Since $p_{i}(\cdot)$ is budget balanced, we obtain $C({\mathcal P'}^{\ast}) = C({\mathcal P}^{\ast})$, where ${\mathcal P'}^{\ast}$ is a utility-based social optimum and ${\mathcal P}^{\ast}$ is a cost-based social optimum. Consider a worst-case stable coalition structure $\hat{\mathcal P}$. Then
\begin{align}
{\sf SPoA}_K^u  \ge & \frac{\sum_{G \in {\mathcal P}^\ast} \sum_{i \in G} u_i(p_{i}(G))}{\sum_{G \in \hat{\mathcal P}} \sum_{i \in G} u_i(p_{i}(G))} \\
 = & \frac{C({\mathcal P}_{\rm self}) - C({\mathcal P}^\ast)}{C({\mathcal P}_{\rm self}) - C(\hat{\mathcal P})} 
 = \displaystyle \frac{\frac{C({\mathcal P}_{\rm self})}{C({\mathcal P}^\ast)} - 1}{\frac{C({\mathcal P}_{\rm self})}{C({\mathcal P}^\ast)} - \frac{C(\hat{\mathcal P})}{C({\mathcal P}^\ast)}} \\
= & \displaystyle \frac{\frac{C({\mathcal P}_{\rm self})}{C({\mathcal P}^\ast)} - 1}{\frac{C({\mathcal P}_{\rm self})}{C({\mathcal P}^\ast)} - {\sf SPoA}_K^C}
\end{align}
Since $\frac{C({\mathcal P}_{\rm self})}{C({\mathcal P}^\ast)} \le K$ (by Lemma~\ref{thm:gen}) and ${\sf SPoA}_K^C \ge 1$, the {\color{black} minimum} is attained when $\frac{C({\mathcal P}_{\rm self})}{C({\mathcal P}^\ast)} = K$.
\end{proof}


\medskip 

Given an order of users ${\mathcal N} = \{i_1, ..., i_n\}$ and cost-sharing mechanism $p_{i}(\cdot)$, we define
\begin{equation}\label{sum}
\alpha\big(p_{i}(\cdot)\big)\triangleq\max_{\stackrel{C(\cdot), G \in {\mathscr P}_K}{\forall s=1,\ldots,k:~p_{i_s}\big(H_s(G)\big)\ge 0}}  \frac{\displaystyle \sum_{s\in\{1,...,n\}:i_s \in G} p_{i_s}\Big(H_s(G)\Big)}{C(G)},
\end{equation}
where 
\begin{equation}
H_s(G) \triangleq \left\{
\begin{array}{ll}
G \backslash \{i_1,...,i_{s-1}\}, & \mbox{\ if\ } i_s \in G, \\
\varnothing,  & \mbox{\ if\ } i_s \notin G\\
\end{array}
\right.
\end{equation}

\bigskip 

In the following, we made a slight but critical modification of a Lemma in \cite{CE17sharing}, which allows us to improve the bound ${\sf SPoA}_K^C$ for egalitarian-split cost-sharing. {\color{black} The crucial observation is that it suffices to consider payment functions which assume non-negative values on the coalitions belonging to a stable coalition structure. Such non-negativity requirement is necessary in the proof of the following lemma and was guaranteed in  \cite{CE17sharing} for  {\it every} coalition structure at the cost of losing a factor of $\sqrt{K}$ in the total cost. However as seen in the proof of the following lemma, it suffices to have non-negativity of payments only on coalitions belonging to a stable coalition structure, a property that is guaranteed to hold of all stable coalition structures under Nash bargaining solution by Lemma~\ref{lem:nashpositive}.}

\begin{lemma}\label{lem:common}  
Consider {\color{black} any }budget balanced cost-sharing mechanism $p_{i}(\cdot)$, which assumes a non-negative value on any coalition belonging to a stable coalition structure. Given a stable coalition structure $\hat{\mathcal P}$ and any coalition structure {\color{black}${{\mathcal P}} \in {\mathscr P}_K$}, we have the upper bound:
\begin{equation}\label{eq:common}
\frac{C(\hat{\mathcal P})}{C({{\mathcal P}})} \le \alpha\big(p_{i}(\cdot)\big)
\end{equation}
\end{lemma}

\begin{proof}
	Let ${\mathcal P} = \{G_1,...,G_h\}$. Define $H_1^1\triangleq G_1$. Then there exists a user $i_1^1\in H_1^1$ and a coalition $\hat G_1^1\in\hat{\mathcal P}_K$, such that $i_1^1\in\hat G_1^1$ and $p_{i_1^1}(H_1^1)\ge p_{i_1^1}(\hat G_1^1)\ge 0$; otherwise, all the users in $H_1^1$ would form a coalition $H_1^1$ to strictly reduce their payments, which contradicts the fact that $\hat{\mathcal P}_K$ is a stable coalition structure.
	
	Next, define $H_2^1\triangleq H_1^1\backslash\{i_1^1\}$. Note that $H_2^1$ is a feasible coalition, because arbitrary coalition structures with at most $K$ users per coalition are allowed in our model.
	By the same argument, there exists $i_2^1\in H_2^1$ and a coalition $\hat G_2^1\in\hat{\mathcal P}_K$, such that $i_2^1\in\hat G_2^1$ and $p_{i_2^1}(H_2^1)\ge p_{i_2^1}(\hat G_2^1)\ge 0$.
	
	Let $G_t = \{i^t_1,...,i^t_{K_t}\}$, for any $t\in\{1,...,h\}$. Continuing this argument, we obtain a collection of sets $\{H_s^t\}$, where each $H_s^t \triangleq\{i^t_s,...,i^t_{K_t}\}$ satisfies the following condition:
	\begin{quote}
		
		{\color{black}For }any $t\in\{1,...,h\}$ and $s\in\{1,..., K_t\}$, there exists $\hat G^t_s\in\hat{\mathcal P}_K$, such that $i_s^t\in\hat G^t_s$ and 
		$p_{i^t_s}(H_s^t)\ge p_{i^t_s}(\hat G^t_s)\ge 0$.
		
	\end{quote}
	Because of budget balance, we have
	\begin{equation}
        \sum_{t=1}^{h} \sum_{s=1}^{K_t}p_{i^t_s}(\hat{G}^t_s) = C(\hat{\mathcal P}_K)
	\end{equation}
	
	Hence, the SPoA, ${\sf SPoA}_K^C$, with respect to $\{p_{i}(\cdot)\}_{i \in {\mathcal N}}$ is upper bounded by
	\begin{align}
	\frac{C(\hat{\mathcal P}_K)}{C({\mathcal P})} &=\frac{\sum_{t=1}^{h} \sum_{s=1}^{K_t}p_{i^t_s}(\hat{G}^t_s)}{\sum_{t=1}^h C(G_t)}\\
	&\le \frac{\sum_{t=1}^{h} \sum_{s=1}^{K_t}p_{i^t_s}(H^t_s)}{\sum_{t=1}^hC(H^t_1)}\\
	&\le\max_{t\in\{1,...,h\}}\frac{ \sum_{s=1}^{K_t}p_{i^t_s}(H^t_s)}{C(H^t_1)}\le \alpha\big(\{p_{i}(\cdot)\}_{i \in {\mathcal N}}\big)
	\end{align}
	because $\alpha(\cdot)$ is non-decreasing in $K$.
\end{proof}

\subsection{Egalitarian-split Cost-Sharing}

Given cost function $C(\cdot)$, we define a truncated cost function $\tilde{C}(\cdot)$ as follows: 
\begin{equation} \label{eqn:truncated}
\tilde{C}(G) \triangleq \left\{
\begin{array}{ll}
C(G), & \mbox{\ if\ } C(G) \le \sum_{j \in G} C_j\\
\sum_{j \in G} C_j, & \mbox{\ if\ } C(G) > \sum_{j \in G} C_j\\
\end{array}
\right.
\end{equation}
Note that $\tilde{C}(G) \le \sum_{j \in G} C_j$ for any $G$.

As mentioned earlier, egalitarian-split cost-sharing is equivalent to Nash bargaining solution. Thus, $${\sf SPoA}^{C,\rm ega}_K = {\sf SPoA}^{C,\rm ns}_K$$

\begin{lemma}[\cite{CE17sharing}] \label{lem:subadd_ega} 
For egalitarian and Nash bargaining solutions, we have
\begin{equation}
 {\sf SPoA}^{C,{\rm ega}}_K (C(\cdot)) =  {\sf SPoA}^{C,{\rm ega}}_K (\tilde{C}(\cdot)) 
\end{equation}
\end{lemma}

\begin{lemma}[\cite{CE17sharing}] \label{lem:nashpositive} 
	For Nash bargaining solution, given a stable coalition structure $\hat{\mathcal P} \in {\mathscr P}_K$ and $G\in \hat{\mathcal P}$, then every user has non-negative payment: $p_i^{\rm ns}(G) \ge 0$ for all $i\in G$. 
\end{lemma}

\begin{lemma}[\cite{CE17sharing}]  \label{lem:nashstb2} 
Let $b_s \triangleq  C(H_s)$.
Consider the following maximization problem:
\begin{align}
& (\textsc{M1}) \ y^\ast(K) \triangleq  \max_{\{C_s, b_s\}_{s=1}^K}  \sum_{s=1}^{K} \Big( C_s - \frac{(\sum_{t=s}^{K}C_t) - b_s}{K-s+1} \Big) \\
& \text{subject to} \notag\\
& \quad b_s\le \sum_{t=s}^KC_t,~\text{for all } s=1,...,K-1, \\
& \quad  0\le C_s\le b_s\le b_{s+1}\le 1,~\text{for all } s=1,...,K, \\
& \quad b_1 + K C_s - \sum_{t=1}^{K}C_t \ge 0,~\text{for all } s =1,...,K \label{con:positive-}
\end{align}
The maximum of (\textsc{M1}) is upper bounded by $y^\ast(K) \le 1 + {\mathcal H}_{K-1}  = O(\log K)$.
\end{lemma}

\begin{customthm}{2} 
For egalitarian-split cost-sharing mechanism and Nash bargaining solution with cost monotonicity, the SPoA is upper bounded by
\begin{equation}
{\sf SPoA}^{C,{\rm ega}}_K  ={\sf SPoA}^{C,\rm ns}_K= O(\log K)
\end{equation} 
\end{customthm}

\begin{proof}
First, by Lemma~\ref{lem:subadd_ega}, it suffices to consider egalitarian bargaining (or Nash bargaining ) solution with {\color{black} a} cost function satisfying $C(G) \le \sum_{j \in G} C_j$ for any $G$. 

Next, by Lemma~\ref{lem:nashpositive}, for any stable coalition structure $\hat{\mathcal P}$, we have $p^{\rm ns}_{i}(G) \ge 0$ for all $i \in G \in \hat{\mathcal P}$.
Hence, by Lemma~\ref{lem:common}, ${\sf SPoA}^{C,{\rm ega}}_K  \le  \alpha\big(\{p^{\rm ega}_{i}(\cdot)\}\big)$.

Let $H_s = \{i_s, ..., i_K \}$, with the default costs denoted by $\{C_s, ..., C_K \}$.
Recall that egalitarian bargaining solution is given by
$$
p^{\rm ega}_{i_s}(H_s) = C_s - \frac{(\sum_{t=s}^{K}C_t) - C(H_s)}{K-s+1} 
$$
subject to $C(H_1) \ge ... \ge C(H_K)$ and $C(H_s) \ge \max\{C_s,...,C_K\}$ (by monotonicity), and $C(H_s) \le \sum_{t =s}^{K} C_t$ (by Lemma~\ref{lem:subadd_ega}).  Finally, it follows that {\color{black}${\sf SPoA}^{\rm eqa}_K = O(\log K)$} by Lemma~\ref{lem:nashstb2}. 
\end{proof}

\subsection{Lower Bound on SPoA}



\begin{customthm}{3} 
Under the assumption that $P\ne NP$, for any cost-sharing mechanism, such that a stable coalition structure can be found in polynomial-time, its SPoA is lower bounded by $\Omega(\log K)$.
\end{customthm}

\begin{proof}
We first model the set cover problem as a coalition formation problem in the following way. Given a collection of sets $\{S_1,...,S_m\}$ of size at most $K$ and the corresponding costs $\{c(S_1),..,c(S_m)\}$, we define a cost function over any subset $G \subseteq \bigcup_{j=1}^m S_j$ as follows:
\begin{equation}
C(G)= \left\{
\begin{array}{cl}
\min_{j \mid G \subseteq S_j} c(S_j), &  \mbox{if there exists $S_j \supseteq G$}\\
  +\infty, & \mbox{otherwise}
\end{array}
\right.
\end{equation}
Note that $C(\cdot)$ obeys cost monotonicity (but $C(\emptyset)>0$). 

We consider coalition formation with the set of participants defined by ${\mathcal N} = \bigcup_{j=1}^m S_j$. Given a social optimum ${\mathcal P}^\ast = \arg\min_{{\mathcal P} \in {\mathscr P}_K} C({\mathcal P})$, it has to be formed by only coalitions that are subsets of the $S_j$'s. Namely, $G \in {\mathcal P}^\ast$, there exists $S_j \supseteq G$.  For each  $G \in {\mathcal P}^\ast$, we take the minimum-cost set $S_j \supseteq G$ and will obtain a set cover ${\mathcal S}$ over ${\mathcal N}$ with exactly the same cost by using $\{S_1, ..., S_m\}$ instead. Namely,
$${\mathcal S} \triangleq \Big\{S_j \mid S_j \supseteq G, C(G) = \min_{j \mid G \subseteq S_j} c(S_j), G \in {\mathcal P}^\ast\Big\}$$  
Note that ${\mathcal S}$ is a minimum-cost set cover over ${\mathcal N}$. Otherwise, ${\mathcal P}^\ast$ is not a social optimum; there is another minimum-cost set cover that can induce a coalition structure with the lower social cost than ${\mathcal P}^\ast$. This shows the reduction of the set cover problem to the coalition formation problem. 

The set cover problem is known to be inapproximable within an approximation ratio better than $\Omega(\log(K))$ in polynomial-time, unless {\sc P=NP}. Hence, in any cost-sharing mechanism, such that a stable coalition structure can be found in polynomial time, its SPoA is lower bounded by $\Omega(\log K)$.
\end{proof}

\subsection{Mixed Cost-Sharing Mechanisms}

\begin{customthm}{4} 
For a mixed cost-sharing mechanism of a constant number of constituent pure cost-sharing mechanisms, if the SPoA of each constituent pure cost-sharing mechanism is $O(f(K))$, then the SPoA of the mixed cost-sharing mechanism is also $O(f(K))$.
\end{customthm}

\begin{proof}
Suppose that the mixed cost-sharing mechanism consists of $M$ constituent pure cost-sharing mechanisms. Let $\hat{\mathcal P}$ be a stable coalition structure induced by the mixed cost-sharing mechanism. We can partition the set of coalitions of $\hat{\mathcal P}$ into $M$ subsets $\{ \hat{\mathcal P}_m\}_{m=1}^M$, where each $G \in \hat{\mathcal P}_m$ is divided by the $m$-th pure cost-sharing mechanism. We write $C(\hat{\mathcal P}_m)$ as the social cost considering the participants in $\hat{\mathcal P}_m$.

Given a subset of participants $x \subseteq {\mathcal N}$,  let $\mbox{\sc Opt}(X)$ be a social optimum considering only $X$. We denote the social cost of $\mbox{\sc Opt}(X)$ by $C(\mbox{\sc Opt}(X))$. Note that $C(\mbox{\sc Opt}(X)) \le C(\mbox{\sc Opt}(Y))$, if $X \subseteq Y$, because the cost function $C(\cdot)$ is monotone and non-negative.

The SPoA of a mixed cost-sharing mechanism can be expressed as $\max_{C(\cdot), \hat{\mathcal P}} \frac{\sum_{m=1}^M C(\hat{\mathcal P}_m)}{C(\mbox{\sc Opt}({\mathcal N}))}$. Next, we can bound the SPoA as follows:

\begin{align}
& \max_{C(\cdot), \hat{\mathcal P}} \frac{\sum_{m=1}^M C(\hat{\mathcal P}_m)}{C(\mbox{\sc Opt}({\mathcal N}))} \\
 \le  & \
\max_{C(\cdot), \hat{\mathcal P}} \frac{\sum_{m=1}^M C(\hat{\mathcal P}_m)}{\tfrac{1}{M}\sum_{m=1}^M C(\mbox{\sc Opt}(\hat{\mathcal P}_m))}  \\
\le & \ 
M \cdot \sum_{m=1}^M \max_{C(\cdot), \hat{\mathcal P}} \frac{ C(\hat{\mathcal P}_m)}{C(\mbox{\sc Opt}(\hat{\mathcal P}_m))} \\
\le & \ M^2 \cdot O(f(K))
\end{align}

Therefore,  the SPoA of the mixed cost-sharing mechanism is also $O(f(K))$, as $M$ is a constant.
\end{proof}

\subsection{Convergence to Stable Coalition Structure}

\begin{customthm}{5} 
If there is no cyclic preference, then {\sc Coln{-}Form} will converge to a stable coalition structure in $O(n^K)$.
\end{customthm}

\begin{proof}
First, we prove that {\sc Coln{-}Form} will terminate in finite time if there is no cyclic preference. We note at each round, one of the participants must make a proposed coalition, which has not been proposed before. Otherwise, the process would have terminated, because every participant is either suspended (i.e., $\mbox{\sc Suspend}_i =\mbox{\sc True}$) or having the proposed coalition not better than the current one held for consideration (i.e., $u_i({\sf H}_i) > u_i(\mbox{\sc Top}(\mbox{\sc Pref}_i))$).  Since every participant $i$'s preference $\mbox{\sc Pref}_i$ is a finite sequence, this process must terminate.

Next, we show that {\sc Coln{-}Form} terminates with a stable coalition structure. We prove by contradiction. Suppose $\tilde{\mathcal P}$ is a coalition structure outcome by {\sc Coln{-}Form}, and $\tilde{\mathcal P}$ is not stable. Namely, there exists a blocking coalition $G_1 \subseteq {\mathcal N}$ with respect to $\tilde{\mathcal P}$, where $|G_1|\le K$, and $u_i(G_1)  >  u_i(G)$ for all $i \in G'$ and $G \in{\mathcal G}_i \cap \tilde{\mathcal P}$.
However, $G_1$ must be a proposed coalition during the process of {\sc Coln{-}Form} and was rejected. Otherwise, $G'$ would be an outcome of {\sc Coln{-}Form}, because all participants in $G_1$ should accept $G_1$, as it is strictly better than the ones in $\tilde{\mathcal P}$. This implies that there exists a participant $i_1 \in G_1$, such that there exists $G_2 \in{\mathcal G}_{i_1}$, which was proposed during the process of {\sc Coln{-}Form} and is more preferred than $G_1$ by $i_1$, (i.e., $u_{i_1}(G_2) > u_{i_1}(G_1)$). 

If $G_2$ is not in $\tilde{\mathcal P}$, then we can similarly find another participant $i_2 \in G_2$, such that there exists $G_3 \in{\mathcal G}_{i_2}$, which was proposed during the process of {\sc Coln{-}Form} and is more preferred than $G_2$ by $i_2$ (i.e., $u_{i_2}(G_3) > u_{i_2}(G_2)$).  We repeat this argument until we find a coalition $G_t \in \tilde{\mathcal P}$ and is more preferred than $G_{t-1}$ by $i_{t-1}$ (i.e., $u_{i_{t-1}}(G_t) > u_{i_{t-1}}(G_{t-1})$. Also, we should be able to find $G_t$ such that $G_t \cap G_1 \ne \varnothing$ because some participants in $G_1$ must be in a coalition in $\tilde{\mathcal P}$. Recall that $\tilde{\mathcal P}$ is not stable, and $G_1$ is a blocking coalition with respect to $\tilde{\mathcal P}$. Namely, there exists $i_t \in G_t \cap G_1$ such that $u_{i_{t}}(G_1) > u_{i_{t}}(G_{t})$.

This creates a cyclic preference, which is a contradiction to the condition of the theorem. Hence, {\sc Coln{-}Form} should terminate with a stable coalition structure. The total running time is $O(n^K)$, because each $|{\mathcal G}_i | = O(n^{K-1})$ and there are total $O(n^K)$ number of coalitions in all preferences of participants. Each coalition must be proposed once only. The total running time is $O(n^K)$. 
\end{proof}

\end{document}